\documentclass{article}
\usepackage{amssymb,amsmath,amsthm,bbm}
\usepackage{verbatim,float,url,dsfont}
\usepackage{graphicx}%,subfigure,psfrag}
\usepackage{algorithm,algorithmic}
\usepackage{mathtools}%,enumitem}
\usepackage{multirow}
\usepackage{ragged2e}
\usepackage{xr-hyper}
\usepackage{array}

\usepackage[colorlinks=true,citecolor=blue,urlcolor=blue,linkcolor=blue]{hyperref}
\usepackage[margin=1in]{geometry}
\usepackage[round]{natbib}

\usepackage[utf8]{inputenc} % allow utf-8 input
\usepackage[T1]{fontenc}    % use 8-bit T1 fonts
\usepackage{booktabs}       % professional-quality tables
\usepackage{nicefrac}         % compact symbols for 1/2, etc.
\usepackage{microtype}      % microtypography

\ifdefined\TimesFont 
\usepackage{times} % use times font
\fi

\ifdefined\ParSkip 
\usepackage{parskip} % use par skip
\fi

\newtheorem{theorem}{Theorem}
\newtheorem{lemma}{Lemma}
\newtheorem{corollary}{Corollary}
\newtheorem{proposition}{Proposition}
\theoremstyle{definition}

\newtheorem*{assumption*}{\assumptionnumber}
\providecommand{\assumptionnumber}{}


\makeatletter
\newcommand*\rel@kern[1]{\kern#1\dimexpr\macc@kerna}
\newcommand*\widebar[1]{%
  \begingroup
  \def\mathaccent##1##2{%
    \rel@kern{0.8}%
    \overline{\rel@kern{-0.8}\macc@nucleus\rel@kern{0.2}}%
    \rel@kern{-0.2}%
  }%
  \macc@depth\@ne
  \let\math@bgroup\@empty \let\math@egroup\macc@set@skewchar
  \mathsurround\z@ \frozen@everymath{\mathgroup\macc@group\relax}%
  \macc@set@skewchar\relax
  \let\mathaccentV\macc@nested@a
  \macc@nested@a\relax111{#1}%
  \endgroup
}
\makeatother

\newcommand{\minimize}{\mathop{\mathrm{minimize}}}

\def\R{\mathbb{R}}

\def\Z{\mathbb{Z}}

\def\E{\mathbb{E}}
\def\P{\mathbb{P}}
\def\T{\mathsf{T}}
\def\Cov{\mathrm{Cov}}
\def\Var{\mathrm{Var}}

\def\th{^{\text{th}}}

\def\hf{\hat{f}}
\def\hmu{\hat{\mu}}

\def\cF{\mathcal{F}}

\def\Risk{\mathrm{Risk}}
\def\Err{\mathrm{Err}}
\def\hErr{\widehat\Err}

\def\UE{\mathrm{UE}}
\def\CB{\mathrm{CB}}
\def\pb{{*b}}
\def\mb{{\dagger{b}}}
\def\Bias{\mathrm{Bias}}
\def\IVar{\mathrm{IVar}}
\def\RVar{\mathrm{RVar}}
\def\Cor{\mathrm{Cor}}
\def\Pois{\mathrm{Pois}}
\def\Binom{\mathrm{Binom}}
\def\Z{\mathbb{Z}}
\def\newY{{\tilde Y}}
\def\newy{{\tilde y}}
\def\dev{\mathrm{dev}}
\def\sqr{\mathrm{sqr}}
\def\UESS{\mathrm{UE}_\mathrm{ss}}
\def\th{^{\text{th}}}

\usepackage[shortlabels]{enumitem}
\usepackage{subcaption}

\title{Unbiased Test Error Estimation in the Poisson Means Problem via Coupled
  Bootstrap Techniques}  
\author{Natalia L.\ Oliveira$^{1,2}$ \and Jing Lei$^{1}$ \and Ryan J.\ Tibshirani$^{1,2}$} 
\date{$^1$Department of Statistics and Data Science, 
  $^2$Machine Learning Department\\ 
    Carnegie Mellon University}

\begin{document}
\maketitle

\begin{abstract} 
We propose a \emph{coupled bootstrap} (CB) method for the test error of an 
arbitrary algorithm that estimates the mean in a Poisson sequence, often called
the Poisson means problem. The idea behind our method is to generate two
carefully-designed data vectors from the original data vector, by using
synthetic binomial noise. One such vector acts as the training sample and the
second acts as the test sample. To stabilize the test error estimate, we average
this over multiple bootstrap $B$ of the synthetic noise. A key property of the
CB estimator is that it is unbiased for the test error in a Poisson problem
where the original mean has been shrunken by a small factor, driven by the
success probability $p$ in the binomial noise. Further, in the limit as $B \to
\infty$ and $p \to 0$, we show that the CB estimator recovers a known unbiased
estimator for test error based on Hudson's lemma, under no assumptions on the
given algorithm for estimating the mean (in particular, no smoothness
assumptions). Our methodology applies to two central loss functions that can be
sused to define test error: Poisson deviance and squared loss. Via a
bias-variance decomposition, for each loss function, we analyze the effects of
the binomial success probability and the number of bootstrap samples and on the
accuracy of the estimator. We also investigate our method empirically across a
variety of settings, using simulated as well as real data.  
\end{abstract}

\section{Introduction}
\label{sec:introduction}

We study the problem of estimating the test error of an algorithm in the Poisson
many means problem, also called the Poisson compound decision problem. The
importance of test error estimation in general rests on the fact that such
estimates can be used in many dowstream applications, such as model assessment,
selection, or tuning. To fix notation, given a data vector $Y = (Y_1,\dots Y_n)
\in \Z_+^n$ (where we write $\Z_+= \{0,1,2,\dots\}$ for the nonnegative
integers) distributed according to  
\begin{equation}
\label{eq:data_model}
\text{$Y_i \sim \Pois(\mu_i)$, \; independently, \; for $i=1,\dots,n$}, 
\end{equation}
we seek to estimate the mean vector $\mu = (\mu_1,\dots,\mu_n) \in \R^n_+$
(where we use $\R_+ = \{ x \in \R : x \geq 0 \}$ for the set of nonnegative 
real numbers). Let $g : \Z_+^n \to \R_+^n$ be a measurable function that
estimates $\mu$ from the data $Y$, so that we can write \smash{$\hmu =
  g(Y)$}. We will often refer to $g$ as an algorithm, in the context of
estimating $\mu$ in the Poisson many means problem. 

To evaluate the performance of $g$, we can use various metrics. One class of
metrics evaluate what we call \emph{test error}, based on a loss function $L :
\Z_+^n \times \R_+^n \to \R$,    
\begin{equation}
\label{eq:test_error}
\text{$\Err(g) = \E[L(\newY, g(Y))]$, \; where $\newY$ is drawn from
  \eqref{eq:data_model}, independently of $Y$}, 
\end{equation}
which measures how well $g$ tracks an independent copy \smash{$\newY$} of the
data. A second class of metrics evaluate what we call \emph{risk}, again based
on a loss function $L$, 
\begin{equation}
\label{eq:risk}
\Risk(g) = \E[L(\mu, g(Y))],
\end{equation}
which measures how well $g$ tracks the mean $\mu = \E[Y]$ of the
data. Admittedly, many authors use the terms ``test error'' and ``risk''
interchangeably, but in this paper we are careful to use terminology that
distinguishes the two, for reasons that we will become apparent in the next  
subsection.      

\subsection{Test error versus risk}
\label{sec:test_error_vs_risk}

In the classical normal means problem, where instead of \eqref{eq:data_model} we
observe $Y_i \sim N(\mu_i, \sigma^2)$ independently, it is straightforward to
show that under a squared loss $L$, the test error \eqref{eq:test_error} and
risk \eqref{eq:risk} differ only by the noise level $\sigma^2$. In the Poisson
means problem, there is no direct analogy for typical loss functions of
interest, and the difference between \eqref{eq:test_error} and \eqref{eq:risk}
will generally depend on $\mu$. This means that an estimator of one metric (test
error or risk) does not as easily translate into an estimator of the other, 
since $\mu$ is of course unknown, and the primary estimand of interest.  

Thankfully, as we show here, when $L$ is a Bregman divergence the difference
between test error and risk does not depend on $g$. A Bregman divergence is a 
loss function of the form $L(a,b) = D_\phi(a,b)$, where
\begin{equation}
\label{eq:bregman}
D_\phi(a,b) = \phi(a) - \phi(b) - \langle \nabla \phi(b), a-b \rangle,
\end{equation}
for a convex, differentiable function $\phi : \R^n \to \R$, where here an
subsequently we use $\langle u, v \rangle = u^\T v$ for vectors $u,v$. In
this case it is straightforward to see that  
\begin{align}
\nonumber
\Err(g) - \Risk(g) 
&= \E[D_\phi(\newY, g(Y))] - \E[D_\phi(\mu, g(Y))] \\
\nonumber
&= \E[\phi(\newY)] - \E[\phi(g(Y))] - 
\E[\langle \nabla \phi(g(Y)), \newY-g(Y) \rangle] \\
\nonumber
&\qquad - \phi(\mu) + \E[\phi(g(Y))] +
\E[\langle \nabla \phi(g(Y)), \mu-g(Y) \rangle] \\ 
\label{eq:jensen_gap}
&= \E[\phi(Y)] - \phi(\mu),
\end{align}
where the cancellation of terms in the third line holds because
\smash{$Y,\newY$} are i.i.d., and thus $\E[\langle \nabla \phi(g(Y)), \newY 
\rangle] = \langle \E[\nabla \phi(g(Y))], \mu \rangle$. Observe that
\eqref{eq:jensen_gap} is the gap in Jensen's inequality. Therefore it is always 
nonnegative, and $\Err(g) \geq \Risk(g)$.  

In this paper, we will focus on estimating the test error \eqref{eq:test_error}
in the Poisson means problem \eqref{eq:data_model}, for two special instances of 
a Bregman divergence: squared loss and Poisson deviance, as will be discussed in 
the next subsection. Since $\Err(g) - \Risk(g) = \E[\phi(Y)] - \phi(\mu)$
depends on $\mu$, it will not be the case that we can automatically translate an
estimator of the test error of $g$ into an estimator of its risk. However, we
can still unbiasedly estimate the \emph{difference} in risk between two models
$g$ and $h$, as discussed next.   

\paragraph{Model comparisons.}

The gap \eqref{eq:jensen_gap} does not depend on $g$. Thus for a comparison 
between two algorithms $g$ and $h$, we always have (provided we use Bregman
divergence to define the test error and risk metrics):
\[
\Err(g) - \Err(h) = \Risk(g) - \Risk(h),
\]
To be clear, this means that if \smash{$\hErr(g)$} is an unbiased estimator of 
$\Err(g)$ for any $g$, just as we will produce in this paper, then
\[
\text{$\hErr(g) - \hErr(h)$ \; is unbiased for \; $\Risk(g) - \Risk(h)$, \; for
  any $g,h$}.
\]
As such, we can still use the tools developed in this paper to perform model
comparisons, or more broadly, model tuning (where $g_s$ is indexed by a tuning  
parameter $s \in S$, and we select $s$ to minimize an unbiased estimate of
test error, or equivalently, risk).

\paragraph{Fixed-X Poisson regression.}

A special case of our problem setting to keep in mind is \emph{fixed-X} Poisson
regression. Here we view $Y \in \R^n$ as a response vector and we have an
associated feature matrix $X \in \R^{n \times p}$. The algorithm $g$ typically
performs a kind of Poisson regression of $Y$ on $X$. As long as we consider $X$
to be fixed (nonrandom), we can still interpret this as a problem of the form
\eqref{eq:data_model}, with $\mu = \mu(X)$. In this setting, the test error
metric \eqref{eq:test_error} translates to what is called fixed-X prediction 
error, where we evaluate predictions at the same feature vectors (rows of $X$),
but against new responses (elements of \smash{$\newY$}).

While fixed-X analyses are more typical in classical statistics, the
\emph{random-X} perspective is great interest in modern prediction
problems. Here the feature vectors at which we make predictions are random,
giving rise to random-X prediction error as the metric of concern. Estimating
random-X prediction error is \emph{not} in general equivalent to estimating
fixed-X prediction error and the two can behave quite differently (see, e.g.,
\citet{rosset2020from} for an extended discussion). The random-X perspective
eludes the framework of the current paper, but is an important topic for future
work.  

\subsection{Squared loss versus Poisson deviance}

When $\phi(x) = \|x\|_2^2$, it is easy to check that 
\[
D_\phi(a,b) = \|a - b\|_2^2, 
\]
which is the squared loss. Meanwhile, when \smash{$\phi(x) = 2\sum_{i=1}^n x_i
  (\log x_i - 1)$}, it follows that
\[
D_\phi(a,b) = 2 \sum_{i=1}^n \bigg( a_i \log \frac{a_i}{b_i} + b_i - a_i \bigg),
\]
which is known as \emph{Poisson deviance}. We will take these to be the two loss
functions of primary interest in our work. Accordingly, we introduce the
notation for test error under squared loss and Poisson deviance: 
\begin{align}
\label{eq:err_sqr}
\Err^\sqr(g) &= \E \|\newY - g(Y)\|_2^2, \\ 
\label{eq:err_dev}
\Err^\dev(g) &= 2 \E \bigg[ \sum_{i=1}^n \bigg( \newY_i \log
  \frac{\newY_i}{g_i(Y)} + g_i(Y) - \newY_i \bigg) \bigg]. 
\end{align} 

Squared loss is a standard choice in many estimation and prediction problems and
does not really need further motivation. Poisson deviance can be motivated from
different perspectives; one nice perspective is that, if we parametrize $g_i(Y)
= \exp(\theta_i)$ for $i=1,\dots,n$, then fitting $g$ to minimize Poisson
deviance on the given data is equivalent to maximum likelihood in the Poisson
model,    
\[
\minimize_g \; 2 \bigg[ \sum_{i=1}^n \bigg( Y_i \log \frac{Y_i}{g_i(Y)} + g_i(Y)
- Y_i \bigg) \bigg] \iff 
\minimize_\theta \; \sum_{i=1}^n \Big( -Y_i \theta_i + \exp(\theta_i) \Big).  
\]
In the same vein, evaluating $g$ by Poisson deviance on \smash{$\newY$} is
equivalent to evaluating $g$ by Poisson likelihood on an independent copy of the 
training sample.       

In our view, squared loss and Poisson deviance are each important loss
functions, and are each deserving of study. This is only strengthened by the
fact that they can have very different behaviors in certain problem settings. As
a simple example, suppose $n=1$, and we have two scenarios: in the first
\smash{$\newY = 1$} and $g(Y) = 2$, while in the second \smash{$\newY = 500$}
and $g(Y) = 501$. The squared loss in each scenario is 1. However, the Poisson  
deviance in first scenario is $\approx 0.307$, and in the second scenario it is 
$\approx 0.001$. The difference here is driven by the fact that in the Poisson
model the variance scales with the mean. Hence according to Poisson deviance
(equivalent to Poisson likelihood), a prediction of 502 when the predictand is
501 is not nearly as bad as a prediction of 2 when the predictand is 1.

In Sections \ref{sec:experiments} and \ref{sec:applications}, we will present
and discuss several examples that expose differences in the behavior of squared
loss and Poisson deviance in different settings. That said, our primary focus is
on estimating test error defined with respect to these loss functions, and not
on comparing them. A comprehensive analysis of their differences is beyond the 
scope of the current paper.    

\subsection{Hudson's lemma}

A fundamental result in this area is \emph{Hudson's lemma}, due to
\citet{hudson1978natural}. Hudson actually derived two identities, one each for   
continuous and discrete exponential families. These can be viewed as extensions  
of Stein's celebrated identity \citep{stein1981estimation} for the Gaussian
family.\footnote{Stein's work was actually completed as a technical report in
  1973, and was a motivation for Hudson's work, even though the publication
  dates of their papers do not reflect this. According to Hudson, Stein already
  knew of the result in \eqref{eq:hudson}.} 
For concreteness, we state Hudson's result for the Poisson case. 

\begin{lemma}[\citealt{hudson1978natural}]
\label{lem:hudson}
Let $Y_i \sim \Pois(\mu_i)$, independently, for $i=1,\dots,n$. Let $g: \Z_+^n
\to \R^n$ be such that $\E|g_i(Y)| < \infty$, $i=1,\dots,n$. Then, denoting by
$e_i \in \R^n$ the vector whose $i\th$ entry is 1, with all others 0, 
\begin{equation}
\label{eq:hudson}
\mu_i \E[g_i(Y)] = \E[Y_ig_i(Y-e_i)], \quad i=1,\dots,n,
\end{equation}
where by convention we set $g_i(-1) = 0$, $i=1,\dots,n$.
\end{lemma}

Compared to Stein's identity, which requires that $g$ is weakly differentiable,
Hudson's identity \eqref{eq:hudson} holds without any smoothness assumptions on
$g$ (of course, even formulating precisely what smoothness would mean over a
discrete domain like $\Z_n^+$ would be tricky, but the lack of assumptions
needed for Lemma \ref{lem:hudson} are remarkable nonetheless). Hudson's main
interest was in developing inadmissibility results for estimators of the
location parameter in an exponential family distribution. The identities 
he established were used as tools in his analysis, which parallels Stein's use 
of his own identity in \citet{stein1981estimation}. 

Moreover, analogous to what can be done with Stein's lemma, Hudson's lemma 
can be used to derived unbiased estimators for various risk metrics in
exponential families. An important contribution in this area is
\citet{eldar2009generalized}, and further contributions (along with a
comprehensive summary of available tools and results from the literature) are
given in \citet{deledalle2017estimation}.           

\subsection{Unbiased estimation}
\label{sec:ue}

Our focus in this paper is slightly unique, since we consider test error
\eqref{eq:test_error} as the primary target and not risk \eqref{eq:risk}, as
considered by \citet{eldar2009generalized, deledalle2017estimation}, and most
other authors in the literature. Nonetheless, the estimators developed by these
authors have natural analogues for test error. In fact, the story is for test
error is simpler, and an unbiased estimator can be obtained for any Bregman
divergence loss function.   

To see this, we first recall a general decomposition of test error for Bregman
divergence losses known as \emph{Efron's optimism theorem}, due to
\citet{efron1975defining, efron1986biased, efron2004estimation}: this shows that
for any Bregman divergence $D_\phi$ in \eqref{eq:bregman} and any algorithm $g$,
this difference in test error and training error satisfies
\begin{align}
\nonumber
\E[D_\phi(\newY, g(Y))] - \E[D_\phi(Y, g(Y))] 
&= \E[\phi(\newY)] - \E[\phi(g(Y))] - 
\E[\langle \nabla \phi(g(Y)), \newY-g(Y) \rangle] \\
\nonumber
&\qquad - \E[\phi(Y)] + \E[\phi(g(Y))] +
\E[\langle \nabla \phi(g(Y)), Y-g(Y) \rangle] \\
\label{eq:opt_bregman}
&= \E[\langle \nabla \phi(g(Y)), Y \rangle] - 
\E[\langle \nabla \phi(g(Y)), \mu \rangle].  
\end{align}
The second line follows from the fact that \smash{$\E[\phi(\newY)] = 
  \E[\phi(Y)]$}.\footnote{This exposes the reason why the analogous
  decomposition for risk can be more complex: when we replace \smash{$\newY$}
  with $\mu$ in the calculation that led to \eqref{eq:opt_bregman}, we are left
  with an extra term \smash{$\phi(\mu) - \E[\phi(Y)]$} that does not cancel and  
  must be estimated.}  
Simply rewriting the above, we see that if we are able to construct an unbiased
estimator $V(g)$ of \smash{$\E[\langle \nabla \phi(g(Y)), \mu \rangle]$} then  
\[
D_\phi(Y, g(Y)) + \langle \nabla \phi(g(Y)), Y \rangle - V(g)
\]
will be an unbiased estimator for $\E[D_\phi(\newY, g(Y))]$. In the Poisson
case, Hudson's identity \eqref{eq:hudson} precisely gives the unbiased
estimator $V(g)$ that we require, which leads to the following result.  

\begin{proposition}
\label{prop:ue_bregman}
Let $Y_i \sim \Pois(\mu_i)$, independently, for $i=1,\dots,n$. Let $g: \Z_+^n
\to \R^n$ be any algorithm and $D_\phi$ be any Bregman divergence loss function
(indexed by a convex, differentiable function $\phi : \R^n \to \R$) such that
$\E|\phi(g(Y))| < \infty$ and $\E|\nabla_i\phi(g(Y))| <\infty$, 
$i=1,\dots,n$. Then    
\begin{equation}
\label{eq:ue_bregman}
\UE(g) = D_\phi(Y, g(Y)) +  \langle \nabla \phi(g(Y)), Y \rangle - 
\langle \nabla \phi(g_-(Y)), Y \rangle
\end{equation}
is unbiased for \smash{$\Err(g) = \E[D_\phi(\newY, g_i(Y))]$}, where we
abbreviate $g_-(Y) = (g_1(Y-e_1), \dots, g_n(Y-e_n))$, and as usual,
\smash{$\newY$} denotes an independent copy of $Y$.     
\end{proposition}

As a consequence, we have the following unbiased estimators for squared loss and
Poisson deviance: 
\begin{align}
\label{eq:ue_sqr}
\UE^\sqr(Y) &= \|Y\|_2^2 + \|g(Y)\|_2^2 - 2 \langle g_-(Y), Y \rangle, \\ 
\label{eq:ue_dev}
\UE^\dev(Y) &= 2 \sum_{i=1}^n \Big( Y_i \log Y_i - Y_i \log g_i(Y-e_i) +
  g_i(Y) - Y_i \Big).  
\end{align}
These are altogether highly similar to the unbiased risk estimators in
\citet{eldar2009generalized, deledalle2017estimation}, and to be clear, we do
not consider \eqref{eq:ue_sqr}, \eqref{eq:ue_dev} to be major (or even original) 
contributions of our work. That said, we have not yet seen the general unbiased
estimator for Bregman divergence \eqref{eq:ue_bregman} noted in the
literature, thus we believe it may be useful to record it (along with the
observation that estimation of test error can be easier than estimation of
risk).  

The estimators in \eqref{eq:ue_sqr}, \eqref{eq:ue_dev} have a clear strength: 
they are unbiased \emph{for any algorithm $g$}. This is a strong property;
recall that by comparison, in the Gaussian model, the analogous estimator is
Stein's unbiased risk estimator (SURE), which requires $g$ to be weakly
differentiable. The estimators in \eqref{eq:ue_sqr}, \eqref{eq:ue_dev} also have
a clear downside: they require the algorithm $g$ to be run $n+1$ times, once to
obtain the original fit $g(Y)$, and then $n$ more times to obtain $g_-(Y)$,
which recall has entries $g_i(Y-e_i)$, $i=1,\dots,n$. Thus we can liken  
\eqref{eq:ue_sqr}, \eqref{eq:ue_dev} to leave-one-out cross-validation, in terms
of computational cost.   

This draws a clear line of motivation to the main contribution of our paper: in
what follows, we develop an unbiased estimator of test error, for any Bregman
divergence loss function $D_\phi$ and any algorithm $g$, using a
carefully-crafted parametric bootstrap scheme. The computational cost (number of
runs of $g$) here is tied to a user-controlled parameter $B$, the number of
bootstrap samples. In general, increasing $B$ decreases the variance of the
estimator, but any choice of $B \geq 1$ yields an estimator that is unbiased for
the test error in a \emph{mean-shrunken} Poisson problem, with mean $(1-p) \mu$,
where $p>0$ is another user-controlled parameter. 

\subsection{Summary of contributions}

The following gives a summary of our main contributions and an outline for this 
paper. 

\begin{itemize}
\item In Section \ref{sec:cb}, we introduce the coupled bootstrap (CB)
  estimator, and prove that it is unbiased for the test error in a mean-shrunken
  Poisson problem.  

\item In Section \ref{sec:noiseless_limit}, we analyze the behavior of the 
  CB estimator as $B \to \infty$ and $p \to 0$, and prove that the limiting
  CB estimator is exactly the unbiased estimator \eqref{eq:ue_bregman} from
  Hudson's lemma. 

\item In Section \ref{sec:bias_variance}, we study the bias and variance of
  the CB estimator and quantify how they depend on $B,p$ and other problem
  parameters. 

\item In Section \ref{sec:experiments}, we compare the CB and the unbiased
  estimator on various simulated data sets, and find that the performance of the
  CB estimator is favorable, especially when the algorithm $g$ is unstable. 

\item In Section \ref{sec:applications}, we examine the use of the CB estimator 
  for model tuning---selecting from a family $g_s$, $s \in S$ of algorithms---in 
  two applications: image denoising and density estimation. We find that using
  Poisson deviance (to define the test error metric) consistently delivers more
  regularized models than using squared loss.     

\item In Section \ref{sec:discussion}, we conclude with a brief discussion and
  ideas for future work.
\end{itemize}

\subsection{Related work}

Estimating risk and test error is of central importance in statistics and
machine learning. In the random-X prediction setting (which recall does not 
fit in the framework of our work) the most ubiquitous estimator is arguably
cross-validation, which itself carries a long line of literature. We do not
describe this literature here, but highlight \citet{bates2021cross} as a nice
recent paper that carefully reexamines this classic estimator, and also provides
a nice overview of literature on cross-validation. 

In the fixed-X prediction setting---or in general, parametric many means   
problems---there has also been a long history of work in statistics, with   
\citet{akaike1973information, mallows1973comments, efron1975defining,
  stein1981estimation, efron1986biased} marking early important
contributions. This has been particularly well-studied in the Gaussian means
problem, and in this area, we draw attention to \citet{breiman1992little,
  ye1998measuring, efron2004estimation}, and particularly to
\citet{oliveira2021unbiased}, as motivation for our current work in the Poisson
means problem. These papers use auxiliary noise---they inject synthetic
(Gaussian) noise into the data at hand---in order to estimate the risk or
fixed-X prediction error of an arbitrary algorithm $g$. Our previous work,
\citet{oliveira2021unbiased}, proposes a coupled bootstrap (CB) scheme for 
doing so that has a simple, intuitive target of estimation for any auxiliary
noise level. In particular, for any auxiliary noise level $\alpha>0$ (a
user-controlled parameter), the CB method produces an unbiased estimator for 
the risk in a Gaussian means problem that has an inflated noise variance
$(1+\alpha) \sigma^2$ (where $\sigma^2$ denotes the original noise level). The 
current paper builds off this idea, and develops a coupled boostrap scheme in
the Poisson model that enjoys analogous properties.      

Relative to the Gaussian case, risk and test error estimation in the Poisson
means model has been less well-studied. However, there has still certainly been
important and influential work in the area. This includes
\citet{hudson1978natural}, as already described in the introduction, and also
\citet{Ye2004, eldar2009generalized, deledalle2017estimation}. Meanwhile,
the literature on \emph{mean estimation}---the role played by what we are
calling the algorithm $g$---in the Poisson many means problem is vast. Quite a
lot of work on this topic has been done in the signal processing community,
where it is often called Poisson denoising; see, e.g., \citet{Harmany2009,
  Luisier2010, Raginsky2010, Harmany2012, Salmon2014, Cao2016}. Therefore we
believe that the techniques we develop for estimating test error in the Poisson
means model should have widespread practical applications, in signal processing
and elsewhere.    

Lastly, we mention a concurrent, related line of work on auxiliary randomization
approaches that allow for rigorous post-selection inference in parametric many
means models, including the Poisson means model. We highlight
\citet{Leiner2021, Neufeld2023} as two nice recent papers in the area. In
particular, a core piece of the auxiliary randomization procedure in our work
was directly inspired by the former paper, and their use of binomial auxiliary
noise in the Poisson model.   

\section{Coupled boostrap estimator}
\label{sec:cb}

In this section, we introduce the CB estimator in the Poisson means model, and
investigate some of its basic properties. 

\subsection{Proposed estimator}

The following is a simple but key ``three-point'' formula for expected Bregman
divergence loss from \citet{oliveira2021unbiased} that will drive our main
proposal in this paper.

\begin{proposition}
\label{prop:three_point}
Let $U,V,W \in \R^n$ be independent random vectors.  For any $g$, and 
Bregman divergence $D_\phi$, 
\begin{equation}
\label{eq:three_point1}
\E[ D_\phi(V, g(U)) ] - \E[ D_\phi(W, g(U)) ] = \E[\phi(V)] - \E[\phi(W)] +
\langle \E[\nabla \phi(U)], \E[W] - \E[V] \rangle,
\end{equation}
assuming all expectations exist and are finite. In particular, if $U,V$ are
i.i.d.\ and $\E[U] = \E[W]$, then
\begin{equation}
\label{eq:three_point2}
\E[ D_\phi(V, g(U)) ]  = \E[ D_\phi(W, g(U)) ]  + \E[\phi(U)] - \E[\phi(W)] . 
\end{equation}
\end{proposition}

\begin{proof}
The first statement \eqref{eq:three_point1} follows from the definition of
Bregman divergence \eqref{eq:bregman}, and the independence of $U,V,W$. The  
second result \eqref{eq:three_point2} follows from the first, by noting that
if $U,V$ are i.i.d.\ and $\E[U] = \E[W]$ then $\E[V] = \E[W]$, thus the last
term on the right-hand side in \eqref{eq:three_point1} is zero, and the first
term is $\E[\phi(U)]$. 
\end{proof}

While simple to state and prove, the results in Proposition
\ref{prop:three_point} are useful observations. To map them onto to the problem
of estimating of test error in the  Poisson model, consider the following. Given 
a Poisson data vector $Y$ from \eqref{eq:data_model}, suppose that we can
generate a pair of vectors \smash{$(U,W)=(Y^*, Y^\dagger)$} that are independent
of each other and have the same mean. Then \eqref{eq:three_point2} says that   
\begin{equation}
\label{eq:cb_idea}
\text{$D_\phi(Y^\dagger, g(Y^*)) + \phi(Y^*) - \phi(Y^\dagger)$ \; is unbiased
  for \; $\E[ D_\phi(\newY^*, g(Y^*)) ]$}, 
\end{equation}
where \smash{$\newY^*$} is an independent copy of $Y^*$. In other words, the
above constructs an unbiased esitmator for the test error in a problem in which
the original data vector was $Y^*$, rather than $Y$. Thus if \emph{$Y^*$ was
  ``close'' in distribution to $Y$}, then this estimator would be meaningful. 
(Ideally, we would like $Y^*$ to be be identical in distribution to $Y$, but
that will not be generically possible without knowledge of $\mu$.)     

What remains is a precise scheme in the Poisson setting to generate the pair
\smash{$(Y^*, Y^\dagger)$} from $Y$ such that \smash{$Y^*, Y^\dagger$} are
independent, \smash{$\E[Y^*] = \E[Y^\dagger]$}, and $Y^*,Y$ are ``close'' in
distribution. The next lemma does the trick and fulfills these three properties
precisely. It was brought to our attention by \citet{Leiner2021} who used it in
a distinct but generally related post-selection inference context. For
completeness, we provide a proof in Appendix \ref{app:pois_binom}. Here and
henceforth we use the following abbreviations: we write $Y \sim \Pois(\mu)$ to 
mean that we draw $Y_i \sim \Pois(\mu_i)$, independently, for $i=1,\dots,n$, and
similarly $Z \sim \Binom(N,  p)$ to mean that we draw $Z_i \sim \Binom(N_i, p)$,
independently, for $i=1,\dots,n$.  

\begin{lemma}
\label{lem:pois_binom}
Given $Y \sim \Pois(\mu)$, fix any $0 < p < 1$ and let $\omega \,|\, Y \sim
\Binom(Y, p)$. Then, defining $Y^* = Y - \omega$ and $Y^\dagger = (1-p)/p \cdot
\omega$, 
it holds that: 
\begin{enumerate}[(i)]
\item $Y^*, Y^\dagger$ are independent;
\item $\E[Y^*] = \E[Y^\dagger]$; and 
\item $Y^* \sim \Pois((1-p) \mu)$.
\end{enumerate}
\end{lemma}

Lemma \ref{lem:pois_binom}, combined with the observation in \eqref{eq:cb_idea},
forms the basis for the CB test error estimator in the Poisson many means
problem. To stabilize the estimator, we can simply repeat the draws of binomial
noise from Lemma \ref{lem:pois_binom} over independent repetitions
$b=1,\dots,B$. To be concrete, this leads to the following method: we first
generate samples according to
\begin{equation}
\label{eq:cb_noise}
\begin{gathered}
\text{$\omega^b \,|\, Y \sim \Binom(Y,p)$, \; 
independently, \; for $b=1,\dots,B$}, \\  
Y^\pb = Y - \omega^b, \quad, 
Y^\mb = \frac{1-p}{p} \omega^b, \quad \text{for $b=1,\dots,B$},
\end{gathered}
\end{equation}
for an arbitrary binomial success probability $0 < p < 1$, and a number of
bootstrap draws $B \geq 1$; then we define the \emph{coupled bootstrap} (CB)
estimator, for test error under Bregman divergence loss $D_\phi$, by: 
\begin{equation}
\label{eq:cb_bregman}
\CB_p(g) = \frac{1}{B} \sum_{b=1}^B \Big( D_\phi(Y^\mb, g(Y^\pb)) + \phi(Y^\pb)
- \phi(Y^\mb) \Big).
\end{equation} 
We can view each \smash{$Y^\pb$} as a synthetic training set for $g$, and each
\smash{$Y^\mb$} as a synthetic test set. The correction term \smash{$\phi(Y^\pb) 
- \phi(Y^\mb)$} accounts for the fact that \smash{$Y^\pb,Y^\mb$} do not have the
same distribution (though recall they do have the same mean, by construction).

For the two loss functions of primary interest, squared loss and Poisson
deviance, the CB estimator in \eqref{eq:cb_bregman} becomes, respectively: 
\begin{align}
\label{eq:cb_sqr}
\CB_p^\sqr &= \frac{1}{B} \sum_{b=1}^B \Big( \|Y^\mb - g(Y^\pb)\|_2^2 +  
 \|Y^\pb\|^2_2 - \|Y^\mb\|_2^2) \Big), \\  
\label{eq:cb_dev}
\CB_p^\dev &= \frac{2}{B} \sum_{b=1}^B \sum_{i=1}^n \Big( Y^\pb_i \log Y^\pb_i
  - Y^\mb_i \log g_i(Y^\pb) + g_i(Y^\pb)- Y^\mb_i \Big).
\end{align}

\paragraph{Interlude: special care with deviance estimators.}

We take a brief but practically important detour to note that special care must
be taken with test error estimators with respect to Poisson deviance loss. In 
this case, each of the unbiased \eqref{eq:ue_dev} and the coupled bootstrap  
\eqref{eq:cb_dev} estimators can diverge if the coordinate functions of $g$   
can output zero. For the unbiased estimator this occurs when $Y_i \not= 0$ and
$g_i(Y-e_i) = 0$; for the coupled bootstrap estimator this occurs when
\smash{$Y^\mb_i \not= 0$} and \smash{$g_i(Y^\pb) = 0$}. As a safety mechanism,
we can simply pad the output of $g$ so that zero is never in the range of its
coordinate functions: say, we can define a modified algorithm 
\[
\tilde{g}_i(y) = g_i(y) 1\{ g_i(y) \not= 0\} + c 1 \{g_i(y) = 0\}, \quad
i=1,\dots,n, 
\]
for a small constant $c>0$. This is reasonable because even the population
Poisson deviance \eqref{eq:err_dev} can itself diverge when the coordinate
functions of $g$ can output zero. With a modified rule like the one above, we
may ask how frequently the padding is actually in effect in the computation of
the estimators \eqref{eq:ue_dev} and \eqref{eq:cb_dev}. We study this in
Appendix \ref{app:diverging_dev} and show that, in a sense, it is typically in
effect less frequently in the CB estimator \eqref{eq:cb_dev} than in the
unbiased one \eqref{eq:ue_dev}.  

\subsection{Unbiasedness for mean-shrunken target}

The next result is immediate from Lemma \ref{lem:pois_binom} and
\eqref{eq:cb_idea}.

\begin{corollary}
\label{cor:cb_bregman}
Let $Y \sim \Pois(\mu)$. Let $g: \Z_+^n \to \R^n$ be any algorithm, let $D_\phi$  
be any Bregman divergence loss function, and let $0 < p < 1$ and $B \geq 1$ be 
arbitrary. Then the CB estimator $\CB_p(g)$ in \eqref{eq:cb_bregman} is 
unbiased for $\Err_p(g)$ (assuming all terms in \eqref{eq:cb_bregman} have
finite expectations), where $\Err_p(g)$ is the test error of $g$ with respect to
a mean-shrunken Poisson problem:    
\begin{equation}
\label{eq:test_error_p}
\text{$\Err_p(g) = \E[ D_\phi(\newY_p, g(Y_p)) ]$, \; where $Y_p,\newY_p \sim
  \Pois((1-p) \mu)$, and $Y_p,\newY_p$ are independent.}  
\end{equation}
\end{corollary}

The strength of Corollary \ref{cor:cb_bregman} rests on the fact that the
estimand in \eqref{eq:cb_bregman} of $\CB_p(g)$ for any choice of $p>0$ is  
highly intuitive: it is $\Err_p(g)$ in \eqref{eq:test_error_p}, which is the
test error that we would encounter in a slightly harder version of our original
problem, where the mean $\mu$ has been replaced by $(1-p) \mu$. 

Why is this important? It means that we do \emph{not} have to send $p \to 0$ in
order to be able to interpret the estimand of the CB estimator, and thus justify
its use. Any nonzero (noninfinitesimal) $p$ will still result in a target that
has a clear, intuitive meaning. This is good news for the CB estimator, because
when $p$ is away from zero, we can generally choose a reasonably small number of
bootstrap draws $B$ in order to stabilize the variance of the estimator, which
presents a computational advantage over the unbiased estimator in
\eqref{eq:ue_bregman}. We will learn more about the behavior of the CB
estimator, as we vary $p$ and $B$, in Sections \ref{sec:bias_variance} and
\ref{sec:experiments} (where we formally analyze the bias and variance, and
carry out empirical comparisons, respectively).

\subsection{Smoothness of mean-shrunken target}

Now that we have shown that $\CB_p(g)$ is unbiased for $\Err_p(g)$, it is
natural to ask whether $\Err_p(g)$ will be close to $\Err(g)$ for small $p$. Our
next result gives a partial answer by proving that if $g$ satisfies some mild 
moment conditions, then the map $p \mapsto \Err_p(g)$ will be continuous (and in 
fact, it can be continuously differentiable, depending on the number of moments  
assumed) in an interval containing $p=0$. Later on, in Section
\ref{sec:bias_variance}, we will derive results that give a more quantitative
sense of how close $\Err_p(g)$ can be to $\Err(g)$.

\begin{proposition}
\label{prop:test_error_p_smoothness}
For $0 \leq p < 1$, let $\Err_p(g)$ be as defined in \eqref{eq:test_error_p}.
If for some integer $k \geq 0$,    
\[
\E\big[ D_\phi(\newY, g(Y)) \langle \newY + Y, 1_n \rangle^m \big] < \infty,
\quad m = 0,\dots,k,
\]
where in the above above \smash{$Y,\newY \sim \Pois(\mu)$} are independent, and
$1_n \in \R^n$ denotes the vector of all 1s, then the map $p \mapsto \Err_p(g)$
has $k$ continuous derivatives on $[0,1)$.
\end{proposition}

The proof of this result is not difficult but a bit technical and deferred to
Appendix \ref{app:test_error_p_smoothness}. We remark that when $k=0$, the
assumption in Proposition \ref{prop:test_error_p_smoothness} is simply
\smash{$\Err(g) = \E[ D_\phi(\newY, g(Y)) ] < \infty$} (i.e., the original test
error is finite), which is extremely weak, and even in this case we get that
$\Err_p(g) \to \Err(g)$ as $p \to 0$.  

\section{Noiseless limit}
\label{sec:noiseless_limit}

In this section, we consider the \emph{infinite-bootstrap} version of the CB
estimator, \smash{$\CB_p^\infty(g) = \lim_{B \to \infty} \CB_p(g)$}. By the law
of large numbers, this estimator is equivalent to taking the expectation over
the binomial noise,   
\begin{equation}
\label{eq:cb_bregman_inf}
\CB_p^\infty(g) = \E[\CB_p(g) \,|\, Y] = 
\E\bigg[ D_\phi\bigg( \frac{1-p}{p} \omega, g(Y-\omega) \bigg) + \phi(Y-\omega)  
- \phi\bigg( \frac{1-p}{p} \omega \bigg) \bigg],
\end{equation}
where $\omega \,| \, Y \sim \Binom(Y,p)$, as in \eqref{eq:cb_noise}.

The next result considers the \emph{noiseless limit} of the infinite-bootstrap
version of the CB estimator \eqref{eq:cb_bregman_inf}, where $p \to 0$. Its
proof is deferred until Appendix \ref{app:noiseless_limit}.

\begin{theorem}
\label{thm:noiseless_limit}
Let $Y \sim \Pois(\mu)$. Let $g: \Z_+^n \to \R^n$ be any algorithm, let
$D_\phi$ be any Bregman divergence loss function, and assume that
$|\phi_i(g(Y))| < \infty$, $|\nabla_i\phi(g(Y))| < \infty$, and
$|\nabla_i\phi(g(Y-e_i))|< \infty$ almost surely, for each $i=1,\dots,n$. Then    
\begin{equation}
\label{eq:cb_noiseless}
\lim_{p \to 0} \CB_p^\infty(g) = \UE(g), \quad \text{almost surely},
\end{equation}
where $\UE(g)$ is the unbiased estimator defined in \eqref{eq:ue_bregman}.
Thus as a consequence, the noiseless limit of \smash{$\CB_p^\infty(g)$} is
unbiased for $\Err(g)$.   
\end{theorem}

That the limiting CB estimator \eqref{eq:cb_bregman} recovers the unbiased
estimator \eqref{eq:ue_bregman} based on Hudson's lemma, as $B \to \infty$ and 
$p \to 0$, is certainly an encouraging property for the former. We recall that
in the Gaussian many means problem, the analogous result was derived in  
\citet{oliveira2021unbiased}: there, the CB estimator recovers the unbiased 
estimator based on Stein's lemma, in the noiseless limit. However, this Gaussian
result requires $g$ to be weakly differentiable (which is the condition required
for Stein's unbiased estimator to be valid in the first place). In the current 
Poisson many means problem, note that Theorem \ref{thm:noiseless_limit} requires
no such restrictions on $g$ (and indeed, recall, the unbiased estimator does not
either, from Proposition \ref{prop:ue_bregman}). 

Figure \ref{fig:noiseless_limit} illustrates this difference via a simple
simulation; see the figure caption for details. 

\begin{figure}[htb]
\centering
\includegraphics[width=\textwidth]{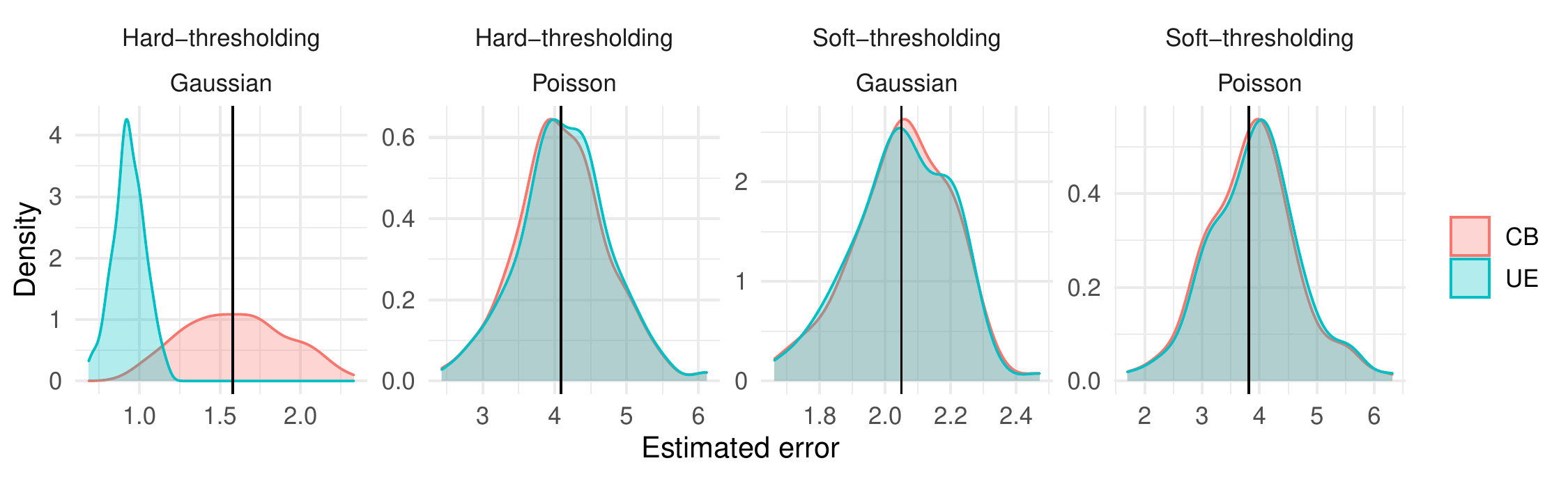}
\caption{Density of the CB and unbiased estimators of test error in Gaussian
  and Poisson settings, where the CB estimators effectively take large $B$ and a
  small amount of auxiliary noise. This is based on a simulation with $n=30$, 
  where we generate Gaussian or Poisson data with constant mean, and consider
  two estimators: soft-thresholding and hard-thresholding (the latter violating
  weak differentiability). The black vertical line in each panel marks the true
  test error. For small auxiliary noise in the Gaussian hard-thresholding
  setting, it is clear that CB and the unbiased estimator are separated (and the
  latter is far from unbiased).} 
\label{fig:noiseless_limit}
\end{figure}

\section{Bias and variance}
\label{sec:bias_variance}

In this section, we study a bias-variance decomposition of the estimator
$\CB_p(g)$ in \eqref{eq:cb_bregman}, when targeting the \emph{original} error
$\Err(g)$. We will consider an arbitrary Bregman divergence loss (used to define 
$\Err(g)$), and use the decomposition     
\begin{equation}
\label{eq:cb_bias_var}
\E[\CB_p(g) - \Err(g)]^2 = 
\underbrace{\big[\Err_p(g) - \Err(g)\big]^2}_{\Bias^2(\CB_p(g))} \,+\, 
 \underbrace{\E\big[\Var(\CB_p(g) \,|\, Y)\big]}_{\RVar(\CB_p(g))} \,+\, 
\underbrace{\Var\big(\E[\CB_p(g) \,|\, Y]\big)}_{\IVar(\CB_p(g))}.
\end{equation}
This is the usual bias-variance decomposition of squared error loss, where we
have used $\E[\CB_p(g)] = \Err_p(g)$ in the bias term, and we have further
expanded the usual variance term (using the law of total variance) into two
components which we call the \emph{reducible} and \emph{irreducible} variance,
respectively, as in \citet{oliveira2021unbiased}. We note that as the number of 
bootstrap draws $B$ grows, the reducible variance shrinks, but the irreducible
variance does not; the latter does not depend on $B$ at all, and in fact, it can
be viewed as the variance of the infinite-bootstrap version of the estimator,
\smash{$\CB_p^\infty(g) = \E[\CB_p(g) \,|\, Y]$}.

In what follows, we will analyze each of the three terms in
\eqref{eq:cb_bias_var} to understand their behavior as functions of $p$ and $B$,
with a focus on small $p$ and large $B$. As usual, we assume throughout that $Y
\sim \Pois(\mu)$, where $Y_p \sim \Pois((1-p) \mu) $ for $p \geq 0$, and we
denote by \smash{$\newY, \newY_p$} independent copies of $Y, Y_p$,
respectively. Lastly, $D_\phi$ represents an arbitrary Bregman divergence loss. 

\subsection{Bias}

First we give an exact expression for the bias, $\Bias(\CB_p(g)) = \Err_p(g) -
\Err(g)$, and an upper bound on its magnitude for small $p$, under an assumption
of monotone variance. The proof is given in Appendix \ref{app:cb_bias}.  

\begin{proposition}
\label{prop:cb_bias}
Assume that \smash{$\E[D_\phi(\newY_p, g_p(Y)) \langle Y_p+\newY_p, 1_n \rangle]
  < \infty$}. Then for all $p \in [0,1)$,   
\begin{equation}
\label{eq:cb_bias} 
\Err_p(g) - \Err(g) = -\sqrt{2 \sum_{i=1}^n\mu_i}
\int_0^p \frac{1}{\sqrt{1-t}} \Cor\Big( D_\phi(\newY_t, g(Y_t)), \langle
\newY_t + Y_t, 1_n \rangle \Big) \sqrt{\Var\big[ D_\phi(\newY_t, g(Y_t)) \big]} 
\, dt.
\end{equation}
Further, if $\Var[D_\phi(\newY_p, g(Y_p))]$ is decreasing in $p$ on $[0,1/2]$,
then for any $p$ in this range,  
\begin{equation}
\label{eq:cb_bias_bd} 
|\Err_p(g) - \Err(g)| \leq \frac{5p}{3} \sqrt{\Var\big[ D_\phi(\newY, g(Y))
  \big] \sum_{i=1}^n\mu_i}
\end{equation}
\end{proposition}

We remark that the assumption of decreasing variance of the loss is fairly
natural (because the variance of each component of $Y_p$ decreases monotonically
to $0$ as $p$ increases to $1$). We can also drop this condition, and replace
the variance term in the bound \eqref{eq:cb_bias_bd} by \smash{$\sup_{t\in[0,p)} 
  \Var[D_\phi(\newY_t,g(Y_t))]$}.   
% One can also obtain a bound of $O(p)$ by invoking the smoothness argument used  
% in Propositions \ref{prop:test_error_p_smoothness}, \ref{prop:cb_rvar}, and 
% \ref{prop:cb_ivar}. But here we choose to present a more explicit constant in
% the bias bound. 

\subsection{Reducible variance}

Next we bound the reducible variance, $\RVar(\CB_p(g))$. We focus on the
dependence on $p$ and $B$, for small $p$ and large $B$. The notation $O(\cdot)$
is to be interpreted in this regime (small $p$, large $B$), and hides factors
that may depend on the mean $\mu$, which may in turn depend on the
dimensionality $n$. The proof is given in Appendix \ref{app:cb_rvar}.   

\begin{proposition}
\label{prop:cb_rvar}
Assume the variables $f(Y_p)$, $f^2(Y_p)$, $f(Y_p) \langle Y_p, 1_n \rangle$, 
$f^2(Y_p)\langle Y_p, 1_n \rangle$ all have finite $L^1$ norm, uniformly bounded  
over all functions $f \in \cF$ and all $p \in [0,q)$, for some $q > 0$, where
\[
\cF = \big\{ y \mapsto D_\phi(y,g(y)) \big\} \cup
\big\{ y \mapsto \nabla_i\phi(g(y)) : i = 1,\dots,n \big\}.
\]
Then for all $p \in [0,q)$, 
\begin{multline}
\label{eq:cb_rvar_bd}
\RVar(\CB_p(g)) \leq \frac{2}{B}\Var\Big[ D_\phi(Y,g(Y))) +\langle 
Y, \nabla \phi(g(Y)) \rangle\Big] + \frac{2}{Bp}\sum_{i=1}^n\mu_i 
\E\big[ \nabla_i\phi(g(Y_p))^2 \big] +{} \\ \frac{2}{B}\sum_{i=1}^n\mu_i^2
\Var\big[ \nabla_i\phi(g(Y_p)) \big] + O\bigg( \frac{p}{B} \bigg). 
\end{multline}
\end{proposition}

A simple simulation, whose results are presented in Figure \ref{fig:cb_rvar},
shows that the reducible variance bound \eqref{eq:cb_rvar_bd} appears to have
the right dependence on $\mu$ and $B$. See the figure caption for details. 

\begin{figure}[htb]
\centering
\begin{subfigure}[b]{0.495\textwidth}
\includegraphics[width=\textwidth]{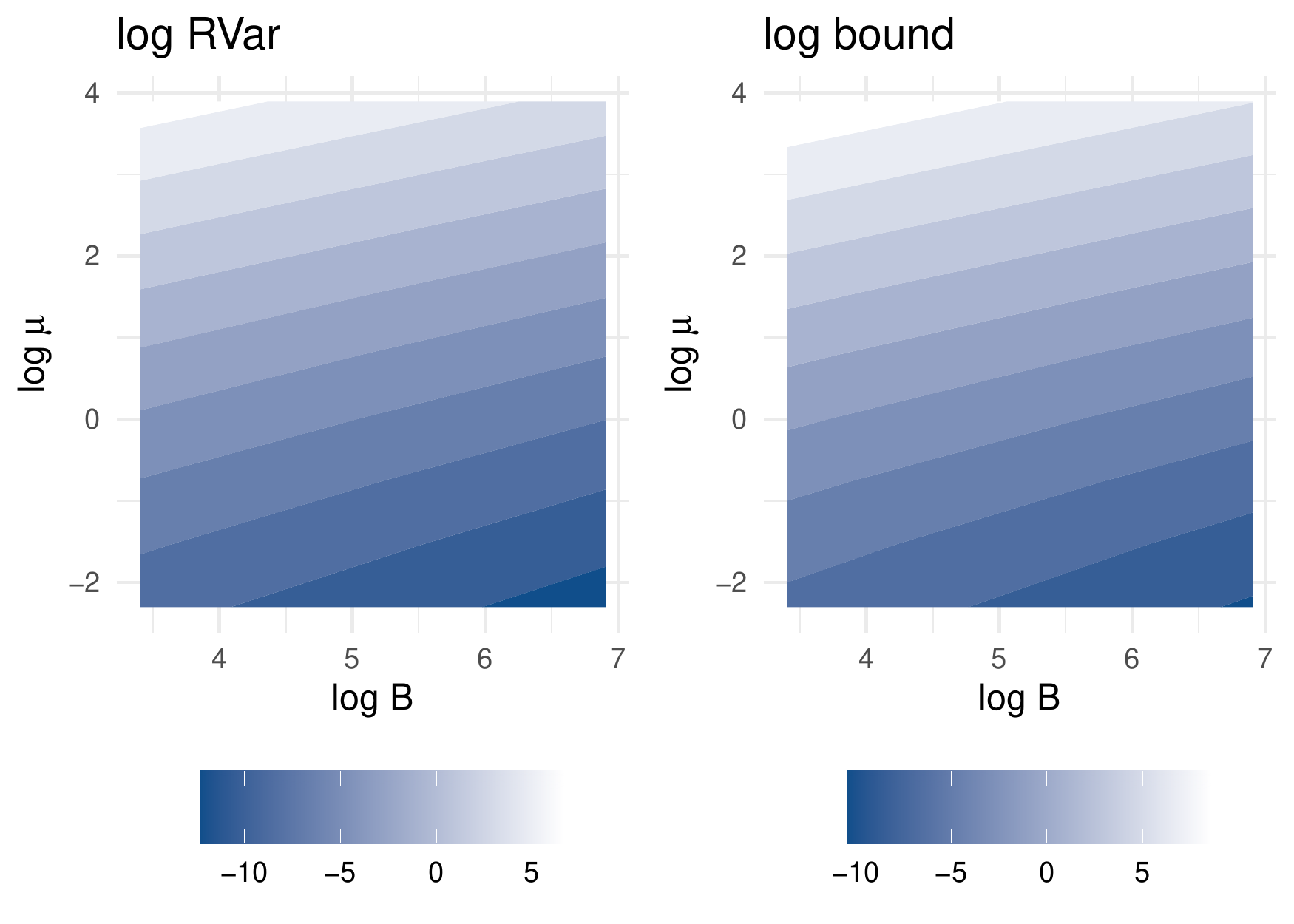}
\caption{Squared loss}
\end{subfigure}
\begin{subfigure}[b]{0.495\textwidth}
\includegraphics[width=\textwidth]{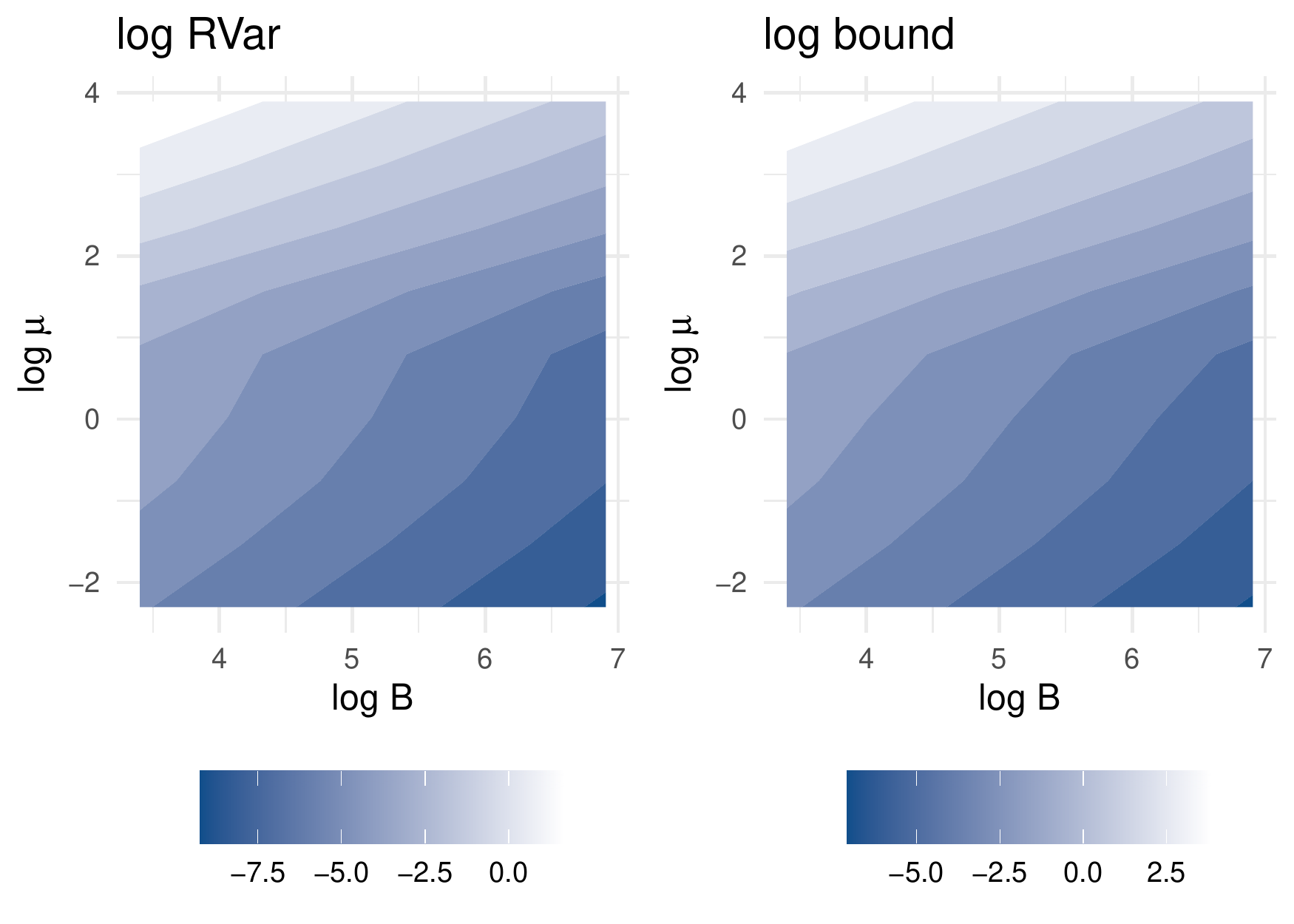}
\caption{Deviance loss}
\end{subfigure}
\caption{Comparison of the true reducible variance (approximated by Monte Carlo)
  and the bound given in \eqref{eq:cb_rvar_bd}, for squared and deviance loss,
  in a simulation with $n=100$ and $p=0.1$. The data vector $Y$ has Poisson 
  entries, and $\mu$ denotes the common mean of each component; we use a simple
  linear shrinkage estimator $g$. We can see that the behavior for varying
  $B,\mu$ looks qualitatively similar across the true reducible variance heatmap
  and the bound heatmap, for each loss function.}      
\label{fig:cb_rvar}
\end{figure}

\subsection{Irreducible variance}

Last we analyze the irreducible variance, $\IVar(\CB_p(g))$. The proof is given
in Appendix \ref{app:cb_ivar}.  

\begin{proposition}
\label{prop:cb_ivar}
Assume that \smash{$\E[D^2_\phi(Y,g(Y))] < \infty$} and
 \smash{$\E[\langle \nabla(g(Y)), Y \rangle^2] < \infty$}. 
Define $\Phi_g$ to have component functions 
\[
\Phi_{g,i}(y) = \sup_{0 \leq z \leq y} \,
\big| \nabla_i\phi(g(y)) \big|, \quad i = 1,\dots,n.
\]
(Here when we write $0 \leq z \leq y$, all inequalities are to be interpreted
componentwise.) Then,     
\begin{equation}
\label{eq:cb_ivar_bd}
\lim_{p \to 0} \IVar(\CB_p(g)) \leq 2 \Var\Big[ D_\phi(Y, g(Y)) + \langle \nabla
\phi(g(Y)), Y \rangle \Big] + 2 \E\big[ \langle \Phi_g(Y), Y \rangle^2 \big].    
\end{equation}
\end{proposition}

\subsection{Discussion of bias and variance results}

We discuss interpretation of the results above. The bias bound
\eqref{eq:cb_bias_bd} decreases linearly with $p$, which suggests that we should
take $p$ to be as small as possible in order to decrease the bias. The
irreducible variance bound \eqref{eq:cb_ivar_bd} provides no resistance to this
idea, as it has a stable noiseless limit, as we send $p \to 0$. The behavior of
the reducible variance bound \eqref{eq:cb_rvar_bd}, however, is more intricate.   
The second term on the right-hand side in \eqref{eq:cb_rvar_bd} diverges as $p
\to 0$, but this can be offset by sending $B \to \infty$. 

How large do we need to take $B$? Altogether, there are really only two
quantities on the right-hand side in \eqref{eq:cb_rvar_bd} that $B$ needs to
balance out, which are the second and third terms. First, let us normalize the
target error by the number of samples, because this would be the natural scale
of concern, in general (our original definition of $\Err(g)$ in
\eqref{eq:err_sqr} or \eqref{eq:err_dev} is a sum, rather than an average, over
samples). We can see from \eqref{eq:cb_bias_var} that rescaling each of
$\Err(g)$ and $\CB_p(g)$ by $1/n$ multiplies all terms in the error
decomposition---bias, reducible variance, and irreducible variance---by a factor
of $1/n^2$. Now, ignoring constants, the (squared) bias bound
\eqref{eq:cb_bias_bd} and the second and third terms in the reducible variance
bound \eqref{eq:cb_rvar_bd} are, after multiplying by $1/n^2$:   
\[
\frac{p^2}{n^2} \|\mu\|_1 \quad \text{and} \quad  
\frac{1}{n^2 Bp} \|\mu\|_1 + \frac{1}{n^2 B} \|\mu\|_2^2,  
\]
respectively, where recall, we use $\mu = (\mu_1,\dots,\mu_n) \in \R^n_+$ for  
mean vector. As we can see, increasing the total signal energy $\|\mu\|_1$
adversely affects the control we have over the bias and reducible variance. In a
moderate signal regime, where $\|\mu\|_1 / n$ is moderate or small, the rough
orders for the bias and the reducible variance in the above display would be
small, even for only modest values of $p$ and $B$. However, in a large signal
regime, where $\|\mu\|_1 / n$ is large (possibly increasing as the sample size
$n$ grows), we may need to take $p$ to be small to offset this (if we want the
bias to be held small), which requires us to take $B$ large enough to dominate
$\|\mu\|_1 / (n^2 p)$ or $\|\mu\|_2^2 / n^2$ (depending on which is larger) in
the reducible variance bound. 

In practice, for any given problem at hand, we would generally recommend
choosing $p$ to be small, such as $p = 0.05$ or $p = 0.1$, but not tiny. This
choice is made in favor of keeping the variance under control (for a reasonable
number of bootstrap samples $B$), at the potential expense of incurring a
nontrivial bias in the CB estimator. However, this brings us back to a primary
feature of the CB estimator---recall, for any $p>0$, it is unbiased for 
$\Err_p(g)$. This represents a shift in focus, where we now consider estimating
error in a problem setting where the mean has been shrunk from $\mu$ to $(1-p)
\mu$, which is intuitively a conservative bet and often a reasonable
undertaking even for moderately small but not infinitesimal values of $p$.   

\section{Simulated experiments}
\label{sec:experiments}

In this section, we run and analyze two sets of simulations. The first,
presented in Section \ref{sec:cb_ue_vp}, compares the unbiased estimator (UE) in
\eqref{eq:ue_bregman} and the CB estimator in \eqref{eq:cb_bregman}, across four
settings. Each setting is defined by a different data model and algorithm $g$,
and we examine the performance of CB versus UE in estimating the true error,
as we vary the binomial noise parameter $p$, for a fixed sample size $n$. We
find that CB performs favorably overall: it delivers similar error estimates to
UE for small values of $p$, and importantly, it can have much smaller variance 
than UE when $g$ is unstable.     

The second set of simulations, presented in Section \ref{sec:cb_ue_ss}, focuses
on just one setting in which UE generally behaves favorably. The motivation
here is to compare the variability of CB and UE after stratifying the two to
have roughly equal computational cost---which is accomplished by sampling
summands in \eqref{eq:ue_sqr} or \eqref{eq:ue_dev}. In this simulation, we fix
the binomial noise parameter $p$, and vary the sample size $n$ and signal size
$\mu$. We find that CB has lower variability unless the signal size $\mu$ is
very large. 
 
\subsection{CB versus UE, varying $p$}
\label{sec:cb_ue_vp}

Here we compare the CB and UE estimators, for squared and deviance loss
functions: see \eqref{eq:cb_sqr}, \eqref{eq:cb_dev} for CB and
\eqref{eq:ue_sqr}, \eqref{eq:ue_dev} for UE. Throughout, we set $n=100$, and use
$B=100$ bootstrap samples for CB. We consider the following combinations of
different data models for $Y$, and algorithms $g$:  
\begin{itemize}
\item \emph{Low-dimensional regression.} We set $p=10$, draw features $X_i \sim
  N(\theta, I_p)$, independently, $i=1,\dots,n$, where each $\theta_j = 3$ and
  $I_p$ denotes the $p \times p$ identity matrix; then we draw responses $Y_i
  \sim \Pois(X_i ^\T \beta)$, independently, $i=1,\dots,n$, where each $\beta_j
  = 0.05$. This corresponds to a signal-to-noise ratio (SNR) of approximately
  2. We examine two choices for $g$, a Poisson regression and a regression tree.  

\item \emph{High-dimensional regression.} We set $p=200$, and use a similar
  setup to the above, except with features $X_i \sim N(0, \sigma^2 I_p)$, where
  $\sigma^2 = 1.5$, and responses $Y_i \sim \Pois(X_i^\T \beta)$, where each
  $\beta_j = 0.13$. This was done to maintain an SNR of roughly 2. In this
  setting, we take $g$ to be a lasso Poisson regression, with the tuning
  parameter $\lambda$ chosen by 5-fold cross-validation (CV).   

\item \emph{Denoising.} We draw $Y_i \sim \Pois(\mu_i)$, independently,  
  $i=1,\dots,n$, where $\mu_i = 10$ for $i \leq 10$ and $\mu_i = 0.5$ for $i >
  10$. In this setting, we take $g$ to be a 1-step improvement on an empirical
  Bayes (EB) estimator as described in \citet{Brown2013}, with tuning parameter
  fixed at $h = 0.85$.    
\end{itemize}
In each setting, we perform 100 repetitions (i.e., we draw the data vector $Y 
\in \R^n$ 100 times from the specified data model); but we note that in the
regression settings, the features are drawn once and fixed throughout. We
consider a range of noise levels $p$ for the CB estimator: 0.05, 0.1, 0.3, 0.5,
and 0.7. All error metrics and error estimators, here and throughout all
empirical examples, are scaled by $1/n$. 

\begin{figure}[tb]
\centering 
\begin{subfigure}[b]{0.485\textwidth}
\includegraphics[width=\textwidth]{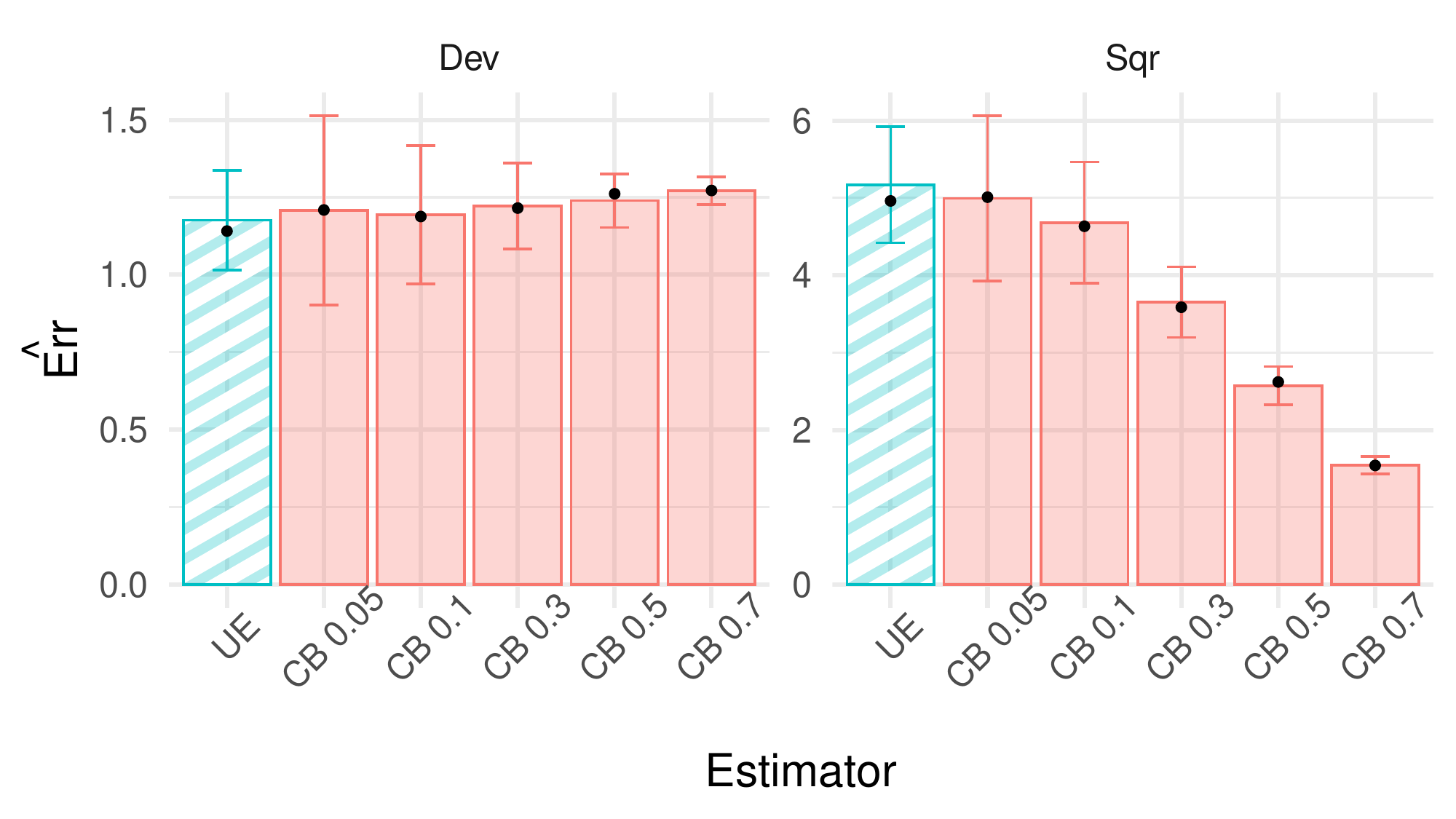}
\caption{Low-dim regression, $g=$ Poisson regression}
\end{subfigure} 
\begin{subfigure}[b]{0.485\textwidth}
\includegraphics[width=\textwidth]{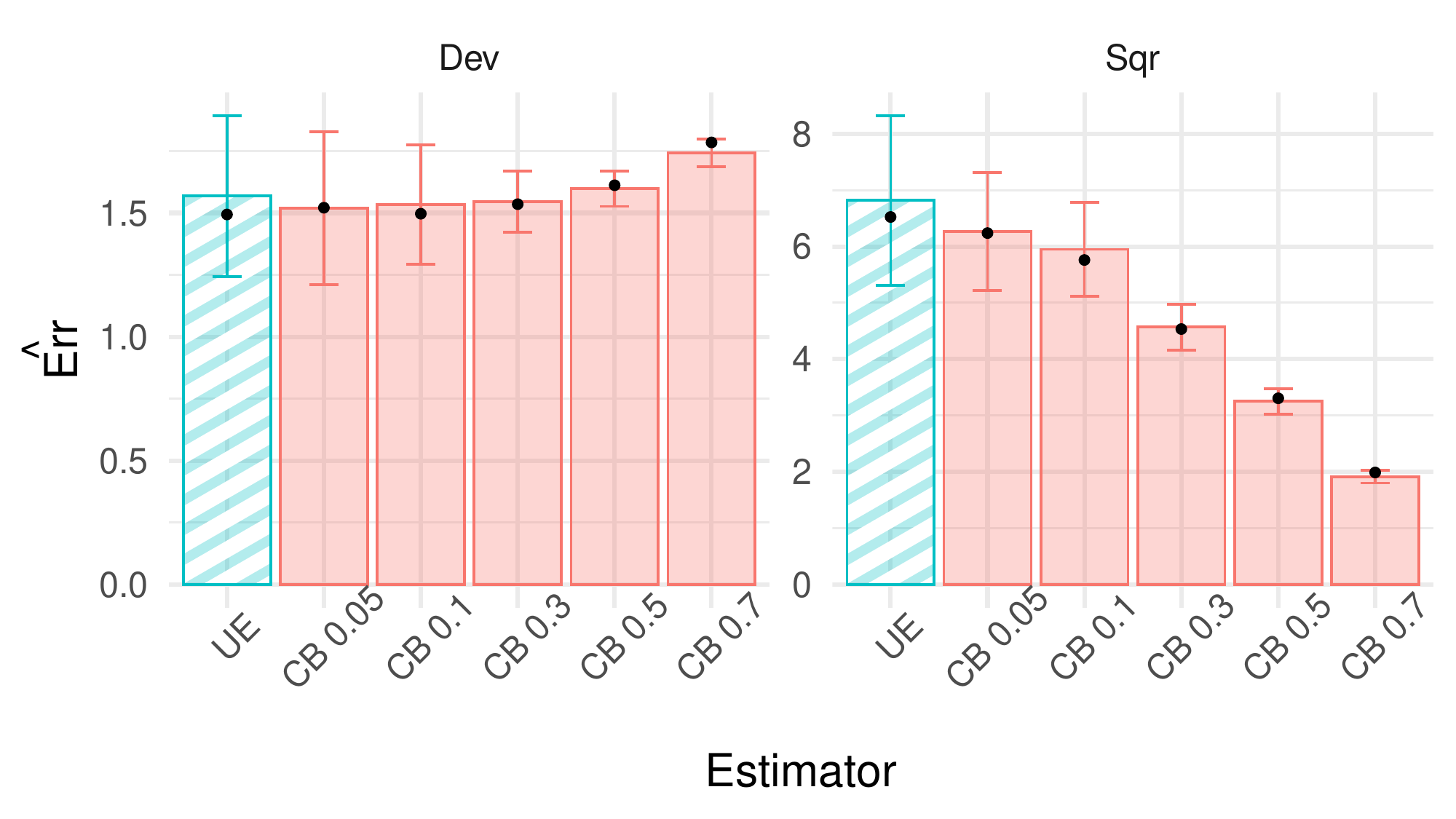}
\caption{Low-dim regression, $g=$ regression tree}
\end{subfigure} 
\begin{subfigure}[b]{0.485\textwidth}
\includegraphics[width=\textwidth]{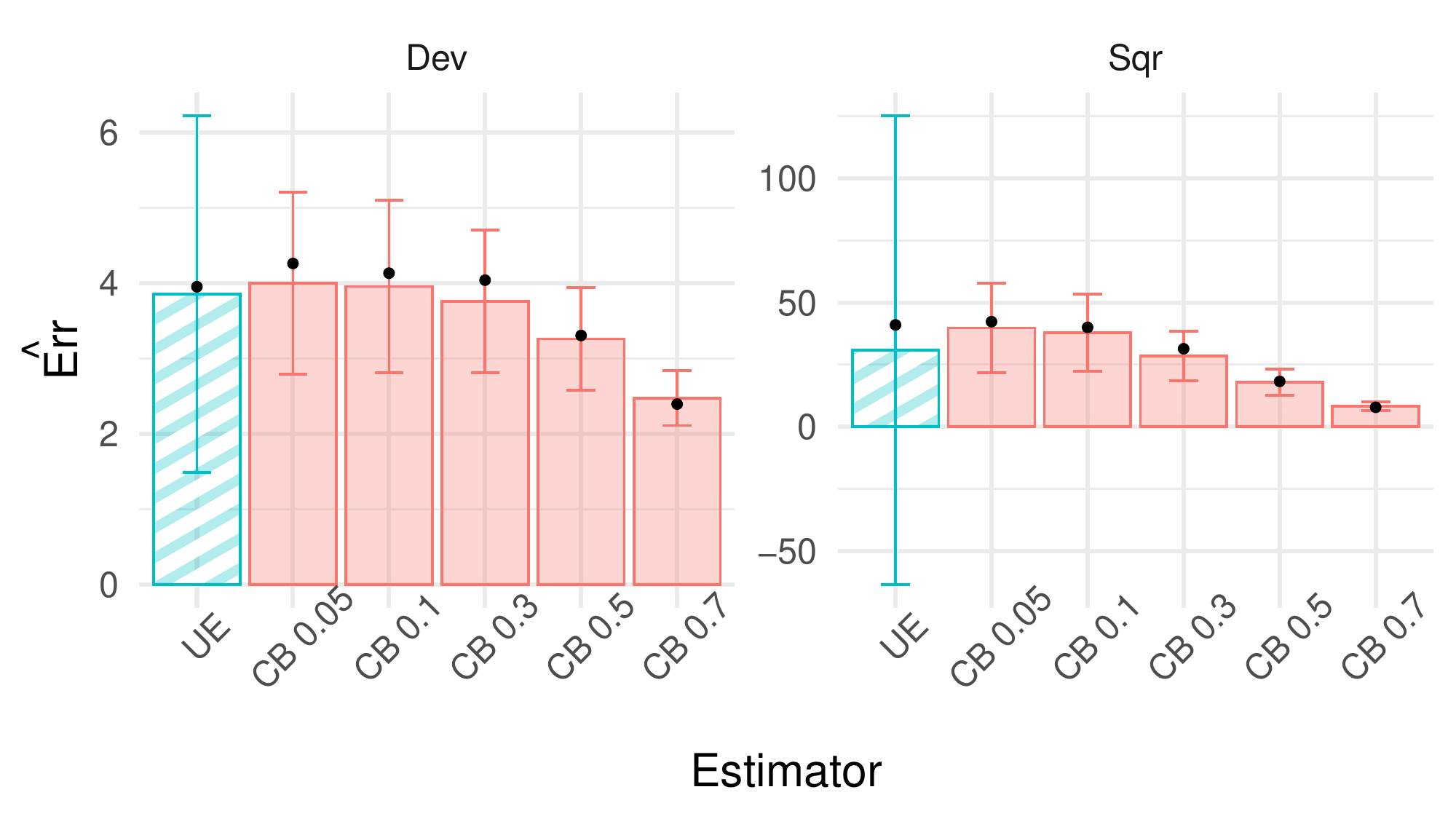}
\caption{High-dim regression, $g=$ 5-fold CV lasso}
\end{subfigure}
\begin{subfigure}[b]{0.485\textwidth}
\includegraphics[width=\textwidth]{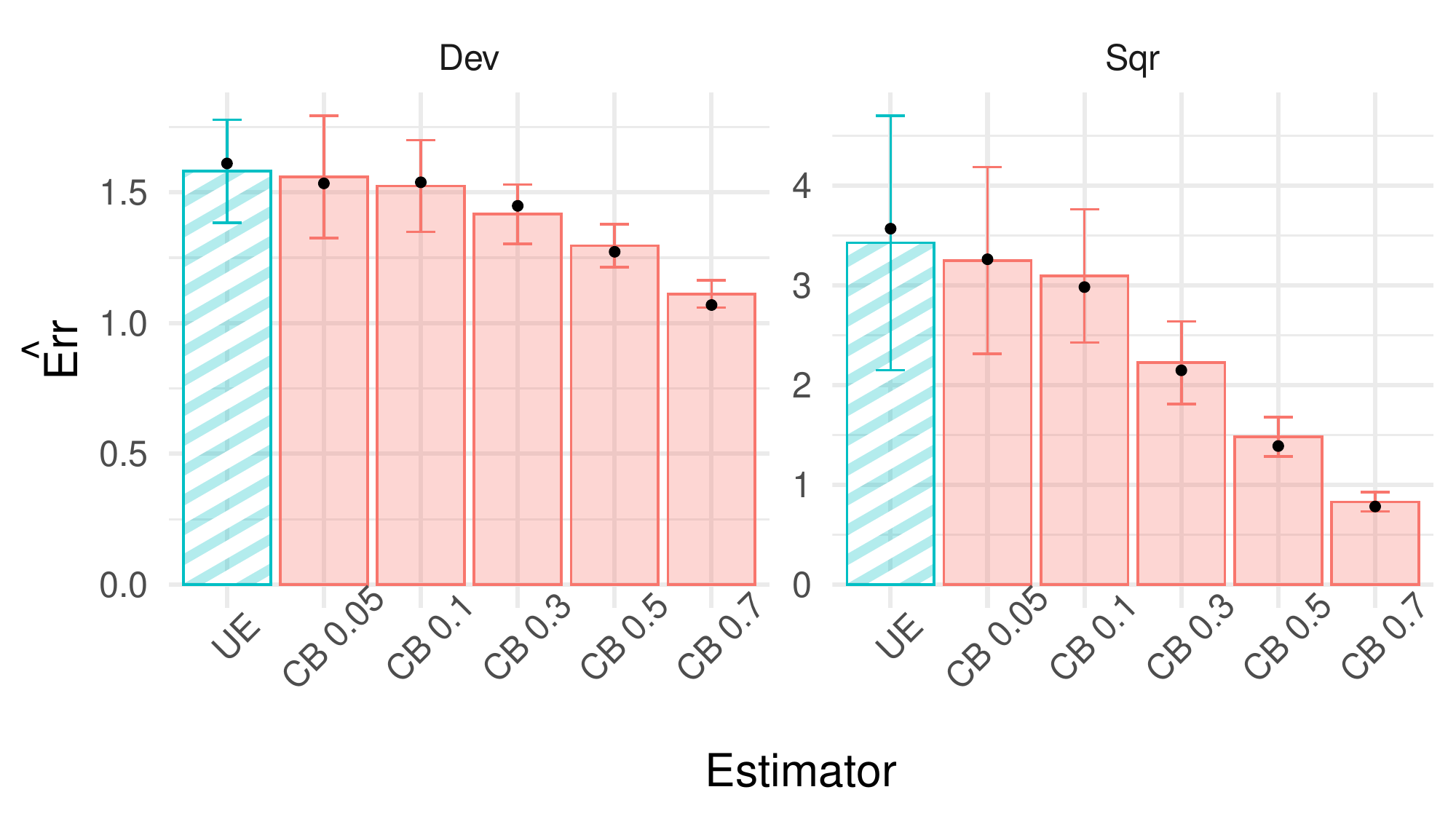}
\caption{Denoising, $g=$ EB 1-step estimator}
\end{subfigure}
\caption{Comparison of CB and UE across different data models, algorithms, and
  loss functions.} 
\label{fig:cb_ue_vp}
\end{figure}

The results are displayed in Figure \ref{fig:cb_ue_vp}, with each panel
(a)--(d) displaying a different combination of data model and algorithm $g$. In
each panel, the average test error estimate is displayed for each method (CB or
UE), as well as standard errors measured over the 100 repetitions. Furthermore,
the black points denote the true estimands (computed via Monte Carlo): $\Err(g)$
for UE, and $\Err_p(g)$ for CB. As expected, all estimators are seen to be
roughly unbiased for their targets: $\Err(g)$ or $\Err_p(g)$. Interestingly, we 
also see that for squared loss, the target $\Err_p(g)$ clearly decreases as $p$
grows, but for deviance loss, the behavior of $\Err_p(g)$ tends to be more
robust to growing $p$ (and can increase or decrease, depending on the setting).  

In panel (a), which is the low-dimensional regression setting, with Poisson
regression as the algorithm $g$, we can see that UE has a noticeably lower
standard error than CB at the lowest binomial noise level of $p=0.05$,
particularly for deviance loss. This is the only setting in which this
happens. In all others, the CB estimator at the lowest noise level has either
comparable or smaller variability than UE. In fact, in panel (c), which is the
high-dimensional regression setting, with CV-tuned lasso as the algorithm $g$,
we see that UE has a dramatically higher standard error than CB at any level of
noise $p$. The algorithm $g$ is inherently unstable here, because CV (operating 
in high-dimensions, and at a moderate SNR) can choose very different tuning  
parameter values across different data instances. Despite this, CB is able to
deliver estimates of reasonably low variance, since it averages across draws of
auxiliary binomial noise, which acts as a method of smoothing (like bagging). We
note that the analogous phenomenon also occurs in the Gaussian setting, as
observed by \citet{oliveira2021unbiased}.

\subsection{CB versus UE, sampling summands}
\label{sec:cb_ue_ss}

The unbiased estimator in \eqref{eq:ue_bregman} requires $n+1$ runs of the
algorithm $g$, making it computationally expensive for large sample sizes. In
contrast, the number of runs of $g$ required by the CB estimator in
\eqref{eq:cb_bregman} is $B$, which is a user-chosen parameter (recall that any 
choice of $B$ results in an unbiased estimator $\CB_p(g)$ for $\Err_p(g)$,
whereas larger $B$ reduces the variance of the estimator). In the last
subsection, we fixed $n = B = 100$. In this one, we consider larger much
sample sizes, with $n$ ranging from $10^3$ to $10^5$. We maintain $B = 100$,
but we equate computational costs between UE and CB by sampling $m = 100$
summands uniformly at random (and without replacement) from
\eqref{eq:ue_sqr} or \eqref{eq:ue_dev}, and then scaling up the resulting sum by 
$n/m$. We denote the estimator resulting from this ``sampling summands''
approach by \smash{$\UESS(g)$}, which is unbiased for $\Err(g)$.    

For the data model, we draw $Y_i \sim \Pois(\mu)$, independently, $i=1,\dots,n$, 
where $\mu$ ranges from 0.5 to 30. For the algorithm, we use a simple linear
shrinkage estimator: \smash{$g(y) = 0.8y + 0.2 \bar y + 0.01 1\{ \bar y = 0
  \}$}. We note that this choice is generally favorable to UE, and more unstable  
algorithms $g$ would only create more variability for UE relative to CB, and
thus look more favorable to CB, as observed in the last subsection. For the
binomial noise parameter, we set \smash{$ p = \min\{ 0.1, \sum_{i=1}^n
  \mu_i/\sum_{i=1}^n \mu_i^2 \}$}, which roughly balances the leading terms in
the reducible variance upper bound \eqref{eq:cb_rvar_bd}.

\begin{figure}[tb]
\centering
\begin{subfigure}[b]{0.875\textwidth}
\includegraphics[width=\textwidth]{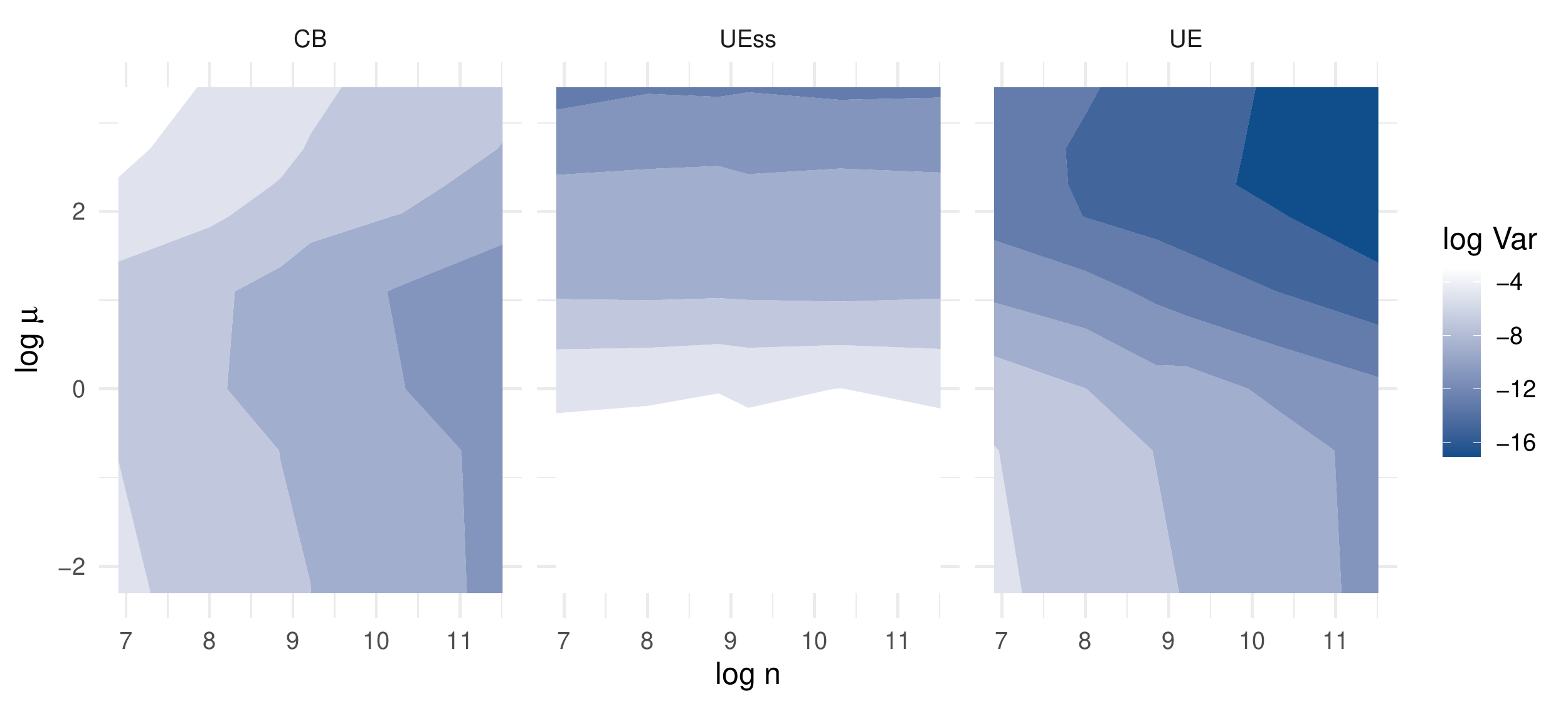}
\caption{Deviance loss}
\end{subfigure}
\begin{subfigure}[b]{0.875\textwidth}
\includegraphics[width=\textwidth]{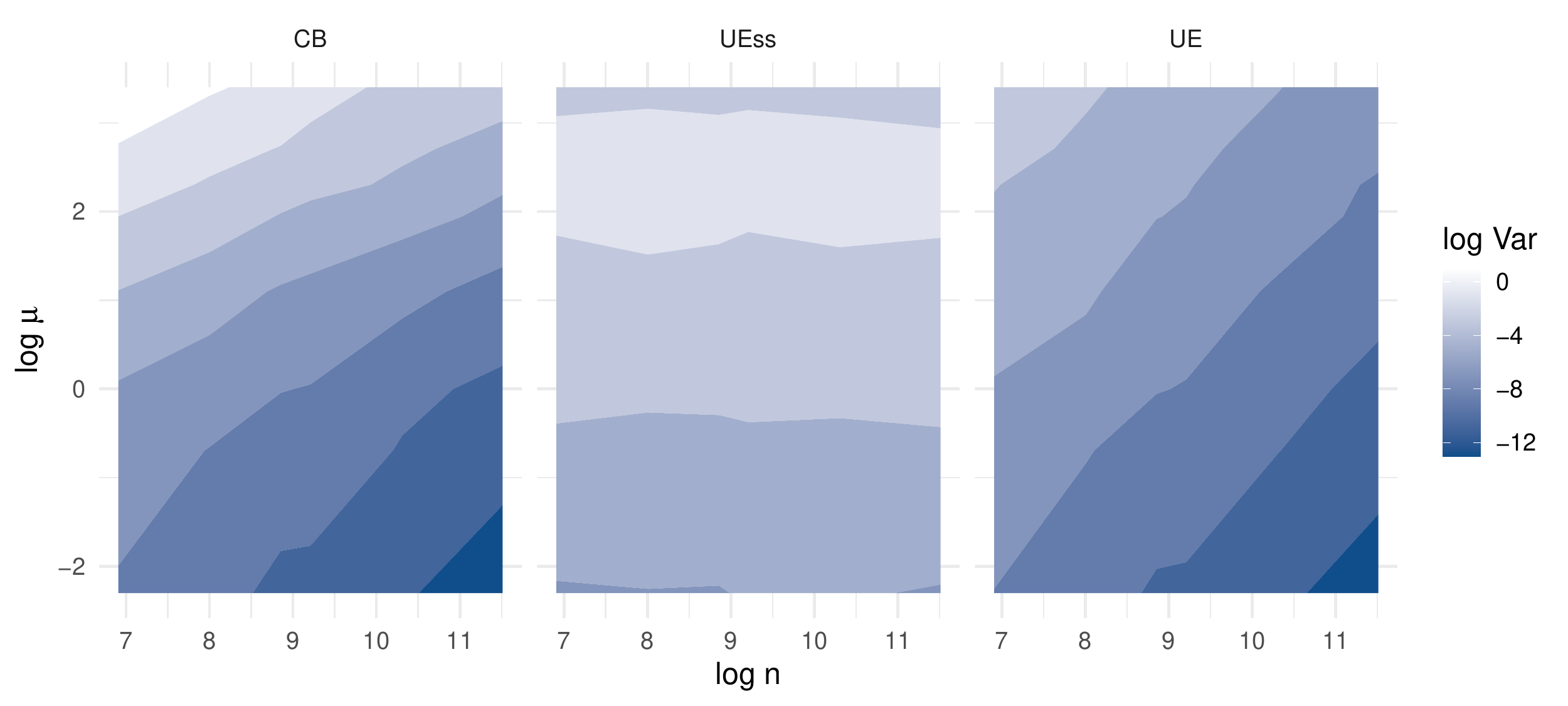}
\caption{Squared loss}
\end{subfigure}
\caption{Variability of CB, \smash{$\UESS$}, and UE as functions of $n$ and 
  $\mu$, for a simple linear shrinkage estimator.} 
\label{fig:cb_ue_ss}
\end{figure}

The results are displayed in Figure \ref{fig:cb_ue_ss}. For deviance loss, the
results are overall quite favorable for CB: it has a lower variance
than \smash{$\UESS$} (darker shade of blue) in all but the top left corner,
which corresponds to small $n$ and large $\mu$. In fact, the variability of CB 
is quite similar (across all $n,\mu$) to that of UE for deviance loss, even
though the former is considerably cheaper ($B=100$ runs of the algorithm $g$,
versus $n$ runs). For squared loss, there is more of a clear tradeoff: for large
values of $\mu$  (i.e., roughly $\log \mu > 2$, or $\mu > 7.38$), we see that CB
has greater variability than \smash{$\UESS$}; for moderate values of $\mu$
(roughly $\log \mu$ between 0 and 2, or $\mu$ between 1 and 7.38), CB has
comparable variability for small $n$ and smaller variability for large $n$;
while for small values of $\mu$ (roughly $\log \mu < 0$, or $\mu < 1$), CB has
smaller variability than \smash{$\UESS$}.  

\section{Applications}
\label{sec:applications}

\subsection{Image denoising}

We consider the following Poisson image denoising framework from
\citet{Harmany2012} (motivated by the study of Poisson noise or shot
noise in areas such as microscopy and astrophotography). We observe data
\smash{$Y_i \sim \Pois(f^*_i)$}, independently, $i=1,\dots,n$, where \smash{$f^* 
  \in \R^n_+$} is an unknown signal of interest, which we assume has the
structure of an $N \times N$ image, where \smash{$n = \sqrt{N}$}. We consider an
estimator \smash{$\hf = g(Y)$} for $f^*$ given by solving the optimization
problem: 
\begin{equation}
\label{eq:tv_denoising}
\minimize_{f \geq 0} \; \sum_{i=1}^n \big(- Y_i \log(f_i + \rho) +  f_i + \rho 
\big) + \tau \sum_{i \sim j} |f_i - f_j|, 
\end{equation}
where $\rho$ is a small positive constant to avoid the singularity $f=0$, and
we write $i \sim j$ to indicate that indices $i,j$ are adjacent to each other in
the ordering determined by the underlying image. This estimator is a form of
total variation (TV) regularized Poisson image denoising; we note that the loss  
term is equivalent (up to constants) to the Poisson deviance between $f$ and
$Y$, and the penalty term encourages the estimated image \smash{$\hf$} to be
piecewise constant, with the tuning parameter $\tau \geq 0$ determining the
strength of regularization.   

As an example, we consider the well-known synthetic phantom image $f^*$, of
resolution $128 \times 128$ (so that $n = 16384$). We use CB to estimate the
test error of \smash{$\hf$}, the Poisson image denoising estimator defined in 
\eqref{eq:tv_denoising}, over a range of values of the tuning parameter
$\tau$. We consider both squared \eqref{eq:cb_sqr} and deviance 
\eqref{eq:cb_dev} loss, and set $B=50$ and $p=0.1$. We did not consider the
unbiased estimator in \eqref{eq:ue_sqr} or \eqref{eq:ue_dev} in this experiment,
due to its prohibitive computational cost (it requires $n = 16384$ refits of the
TV denoising estimator \eqref{eq:tv_denoising}).     

\begin{figure}[htb]
\centering
\includegraphics[width=\textwidth]{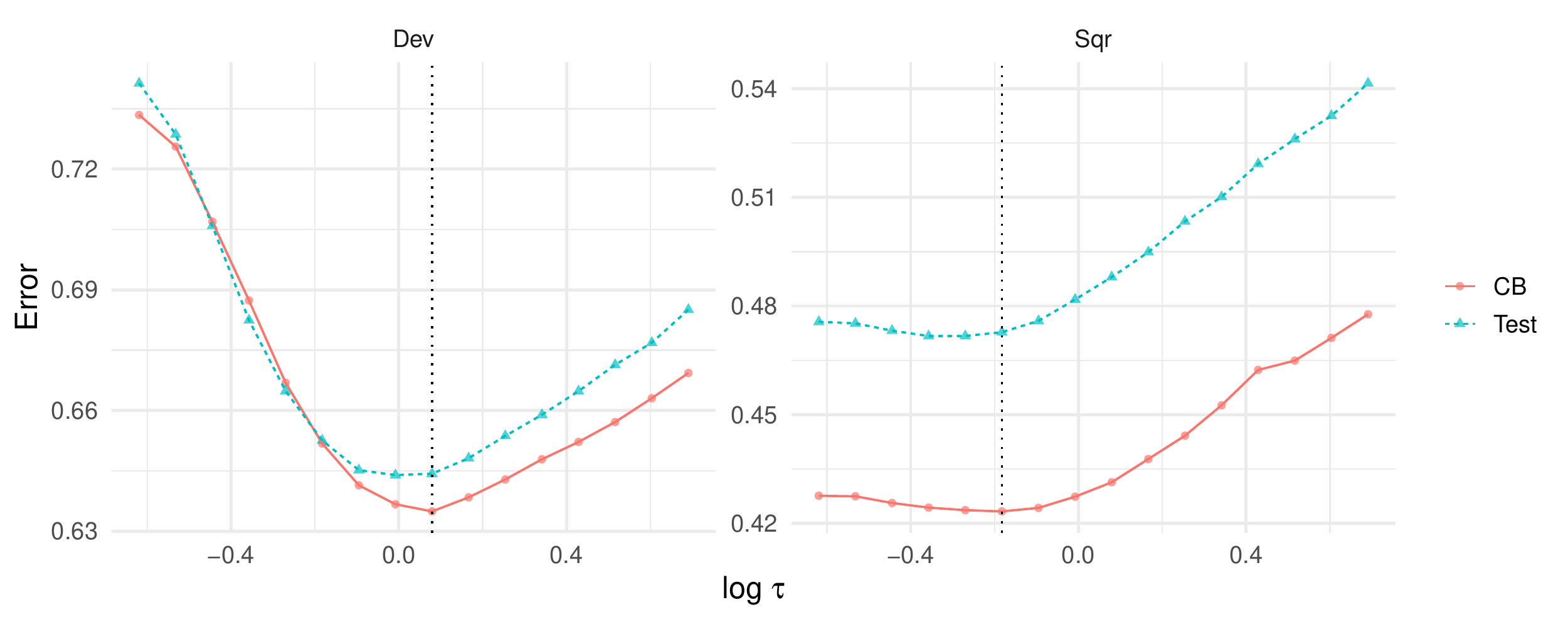}
\caption{Comparison of CB and true test error curves, as functions of the 
  tuning parameter $\tau$, for a Poisson image denoising estimator.}  
\label{fig:phantom_error}
\end{figure}

\begin{figure}[htb]
\centering
\includegraphics[width = \textwidth]{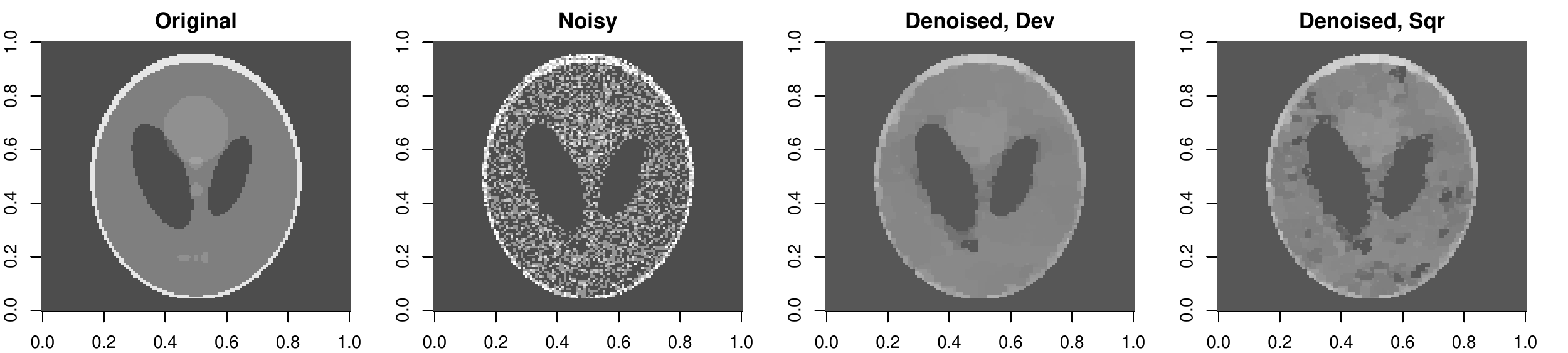}
\caption{Original, noisy, and denoised estimates at the CB-optimal value of
  $\tau$, for squared and deviance loss.}
\label{fig:phantom_image}
\end{figure}

Figure \ref{fig:phantom_error} displays the CB error curve and true error curve
(approximated by Monte Carlo), as functions of $\tau$, with separate panels for
squared and deviance loss. There are two points worth noting. First, despite the
gap between the CB and true test error curves (unsurprising, because $p=0.1$),
their curvature is similar; in particular, the value of $\tau$ minimizing the CB
curve is close to the value minimizing test error. This is the case for both
squared and deviance loss, and it shows that the CB estimator can be useful for
model tuning, even when $p$ is not small. Second, the value of $\tau$ minimizing
the CB curve is larger for deviance loss than it is for squared loss, marked by
the dotted lines in each panel. This translates into a greater degree of
regularization, as can be seen clearly in Figure \ref{fig:phantom_image}, which
plots the denoised estimates themselves at the CB-optimal values of $\tau$, for 
squared and deviance loss.

\subsection{Density estimation}
\label{sec:density}

We study density estimation, which can be turned into a Poisson regression via
Lindsey's method \citep{1974LindseyProb, Lindsey1992, EfronTibs1996}. The basic
idea is to discretize the domain into bins, and model the count in each bin as a
Poisson random variable, with the mean parameter constrained or regularized to
be smoothly varying across the bins. We can then estimate the mean parameter 
by (regularized) maximum likelihood, which in turn gives a discretized density
estimate.  

In particular, we consider an example from \citet{Phillips2006} on the
distribution of \emph{Bradypus variegatus}, a lowland species of sloth found
across Central and South America. Each data point consists of latitude and
longitude pair, representing a site where a sloth was seen, and the data set
contains 116 total sightings. To form a 2d density estimate, we apply Lindsey's
method, with 200 equally-spaced bins along the latitude and longitude axes, and
we use a P-spline to model the Poisson mean parameter. P-splines were first
proposed by \citet{Eilers1996}; details for the 2d case can be found in
\citet{Eilers2003, Eilers2006, Eilers2021}. We use a 2d cubic B-spline
parametrization for the mean function with 30 knots in each dimension, and we
use a penalty on the sum of squared second-order differences across adjacent 
B-spline parameters along each dimension. Moreover, we consider two versions of
this penalty: an \emph{anisotropic} version, which decouples the regularization
strength along each dimension, and has two tuning parameters $\lambda_1,
\lambda_2 \geq 0$; and an \emph{isotropic} version, which ties together the
regularization strength over the dimensions, and has a single tuning parameter  
$\lambda \geq 0$.

\begin{figure}[htbp]
\centering
\begin{subfigure}[b]{0.4\textwidth}
\includegraphics[width=\textwidth]{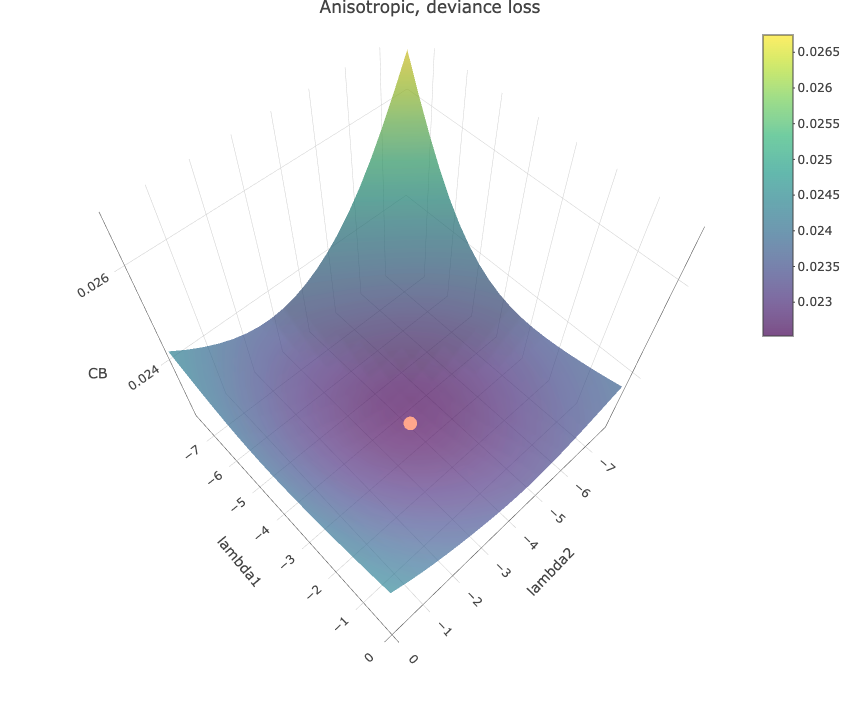}
\caption{Anisotropic penalty, deviance loss}
\end{subfigure}
\begin{subfigure}[b]{0.4\textwidth}
\includegraphics[width=\textwidth]{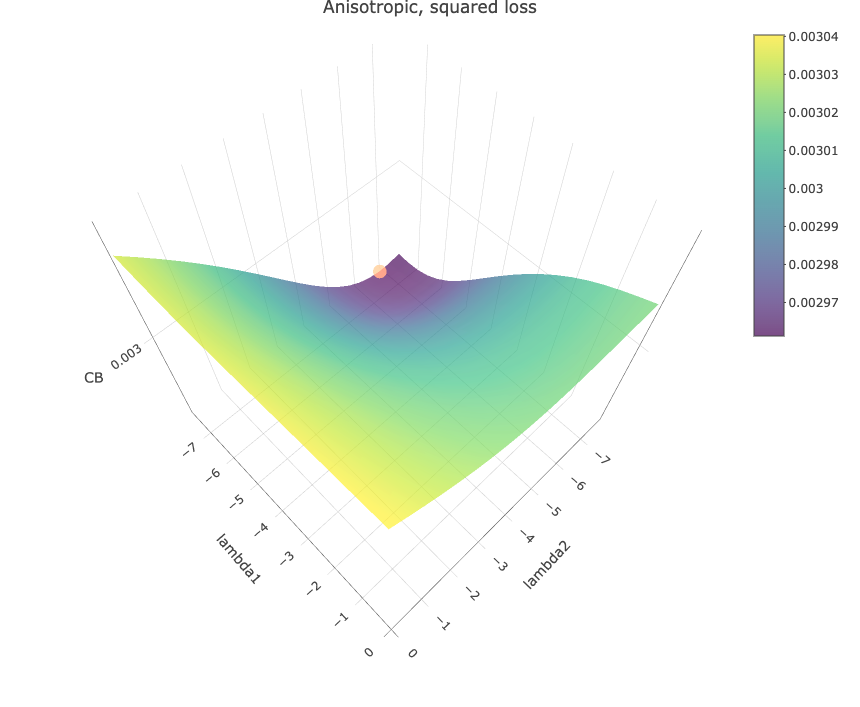}
\caption{Anisotropic penalty, squared loss}
\end{subfigure}
\begin{subfigure}[b]{0.8\textwidth}
\centering
\includegraphics[width=\textwidth]{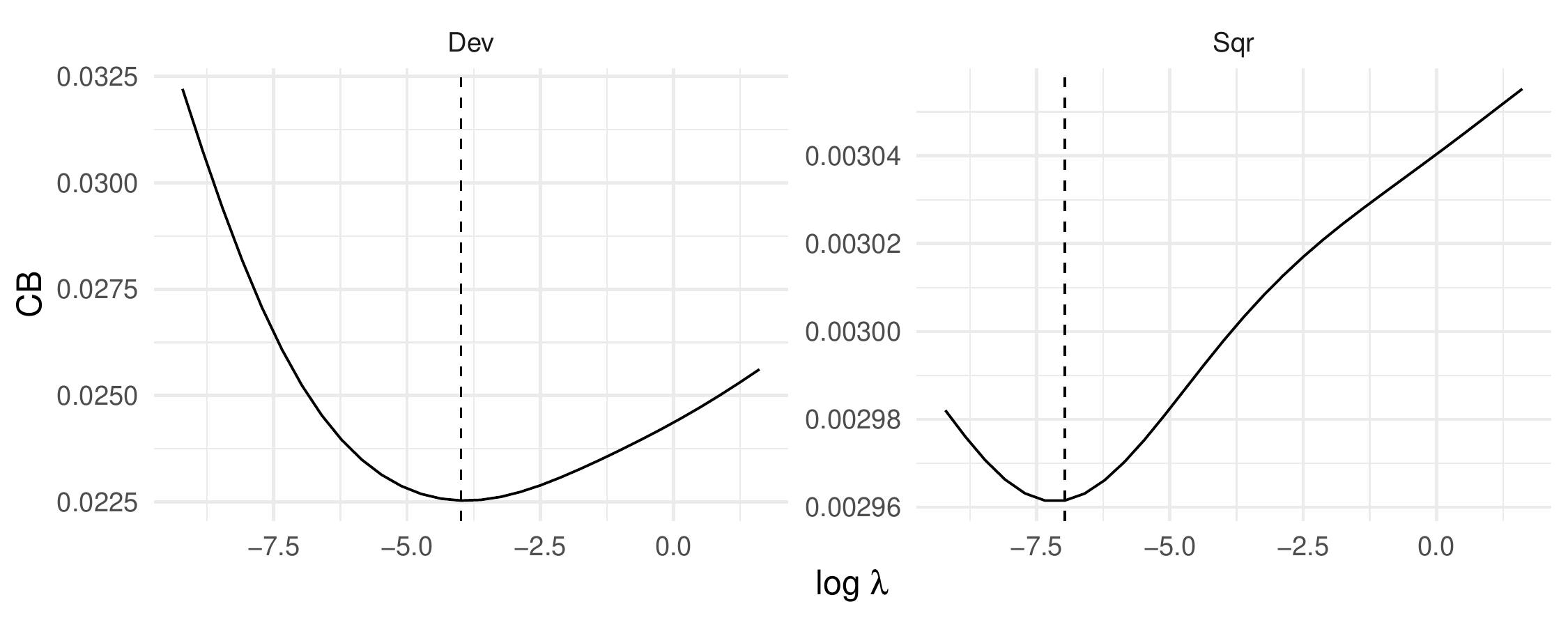}
\caption{Isotropic penalty}
\end{subfigure}
\caption{CB curves for anisotropic and isotropic penalties as a function of the
  tuning parameter(s).}
\label{fig:sloth_error}

\bigskip
\includegraphics[width = 0.7\textwidth]{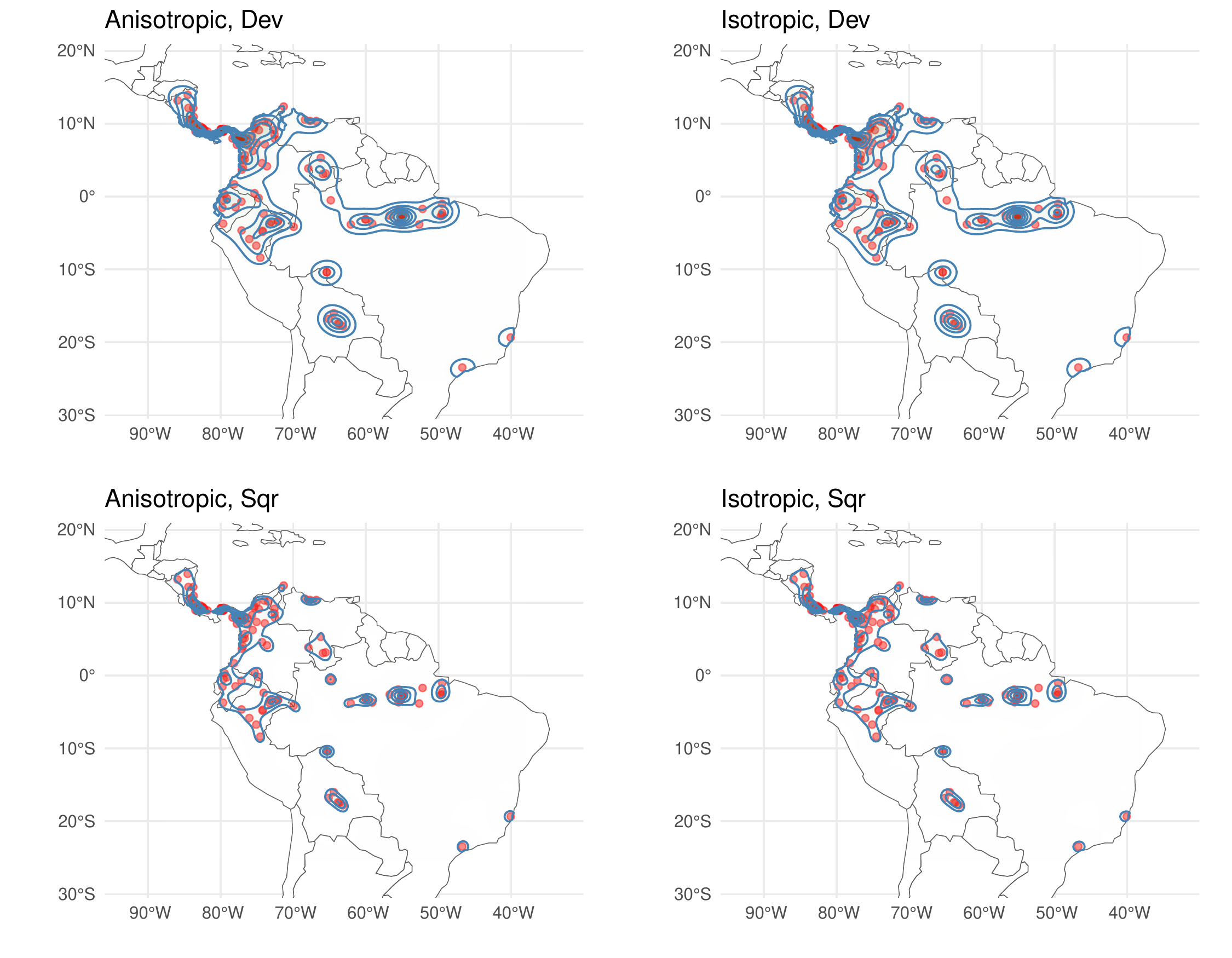}
\caption{Density estimates at the CB-optimized values of the tuning parameters
  for anisotropic and isotropic penalties, and squared and deviance loss.}
\label{fig:sloth_density}
\end{figure}

Figure \ref{fig:sloth_error} shows the results of using the CB method, with
$B=100$ and $p=0.1$, to estimate both squared and deviance loss, across a range 
of tuning parameter values. For either the anisotropic or isotropic penalty, it
is clear that minimizing CB-estimated deviance loss leads to larger tuning
parameter values---and hence more regularized density estimates---than
minimizing CB-estimated squared loss. As we can see in Figure
\ref{fig:sloth_density}, this leads to more plausible looking density estimates
(top row). The estimates obtained optimizing CB-estimated squared loss (bottom
row) appear too concentrated around the observations themselves.  

\section{Discussion}
\label{sec:discussion}

We proposed and analyzed a coupled bootstrap (CB) method for test error
estimation in the Poisson means problem, with a focus on squared and Poisson
deviance loss functions. The CB estimator, for any choice of the binomial noise
parameter $p>0$, is unbiased for an intuitive target: $\Err_p(g)$, the test 
error of the given algorithm $g$, when the mean vector in the Poisson
model has been shrunk from $\mu$ to $(1-p) \mu$. Importantly, this unbiasedness
requires no assumptions on $g$ whatsoever. Furthermore, we proved that in the
noiseless limit $p \to 0$, the CB estimator (with infinite bootstrap iterations)
reduces to the natural unbiased estimator (UE) for test error that comes from an
application of Hudson's lemma. However, CB has two key advantages over
UE. First, it requires running the algorithm $g$ in question $B$ times (which is
a user-controlled parameter in CB), versus $n+1$ times (which comes directly
from the form of UE). Second, as we show in our experiments, CB can often have
smaller variance than UE, particularly when the underlying algorithm $g$ is
unstable.

We finish by emphasizing that it would be interesting to extend the CB framework
to other data models, beyond Gaussian, as in \citet{oliveira2021unbiased}, and 
Poisson, as in the current paper. To explain what would be required for this, it
may be helpful to first recap the general developments in Section
\ref{sec:cb}. Given any random vector $Y$, suppose that we can generate a pair
$(Y^*, Y^\dagger)$ such that: 
\begin{enumerate}[(i)]
\item $Y^*, Y^\dagger$ are independent; and 
\item $\E[Y^*] = \E[Y^\dagger]$.
\end{enumerate}
Then letting \smash{$\newY^*$} denote an independent copy of $Y^*$, Proposition  
\ref{prop:three_point} implies (as stated in \eqref{eq:cb_idea}, which we copy
here for convenience): 
\[
\text{$D_\phi(Y^\dagger, g(Y^*)) + \phi(Y^*) - \phi(Y^\dagger)$ \; is unbiased
  for \; $\E[ D_\phi(\newY^*, g(Y^*)) ]$}, 
\]
for any Bregman divergence $D_\phi$ which serves as our loss
function. Therefore, under properties (i) and (ii) we can estimate error as
measured by an arbitrary Bregman divergence, unbiasedly---granted, the error
here is defined when the training and test distributions are given by that of
$Y^*$, which is different from our original data distribution. This means that
there is actually an implicit third property that we need in order for us to
want to use the estimator in the above display:   
\begin{enumerate}
\item[(iii)] the law of $Y^*$ is ``close enough'' to that of $Y$ that
  \smash{$\E[D_\phi(\newY^*, g(Y^*))]$} is an ``interesting'' proxy target.  
\end{enumerate}
This is less explicit than either (i) or (ii) but it is just as important. To be
clear, the properties (i), (ii), and (iii) are already met by the existing
Gaussian and Poisson constructions. Moreover, for any given distribution of
$Y$, if we can fulfill (i), (ii), and (iii), then we can build a corresponding
CB estimator for the (proxy) test error by averaging the above construction over  
multiple bootstrap draws, as in \eqref{eq:cb_bregman}.            

Towards satisfying properties (i) and (ii), the recent paper of
\citet{Neufeld2023} provides a number of constructions which serve a related but  
distinct purpose, in a selective inference context. From some initial random vector $Y$, 
they seek to create a pair \smash{$(Y^{(1)}, Y^{(2)})$} which are independent,
and satisfy \smash{$Y^{(1)} + Y^{(2)} = Y$}. Fortunately, by simple rescaling,
one can check that their constructions (from their Table 2) can be adapted
to satisfy (i) and (iii) for the gamma, exponential, binomial, multinomial, and
negative binomial families of distributions. Meanwhile, property (iii) can be 
argued on a case-by-case basis. As an example, consider $n=1$ (only for
simplicity, the same idea can be applied coordinatewise in the multivariate 
case), and assume $Y \sim \mathrm{Exp}(\lambda)$, exponentially distributed with 
rate $\lambda > 0$. Then for arbitrary $\epsilon \in (0,1)$, we can define    
\begin{align*}
Z &\sim \mathrm{Beta}(\epsilon, 1-\epsilon), \\
Y^* &= \frac{Z}{\epsilon} \cdot Y, \\
Y^\dagger &= \frac{1-Z}{1-\epsilon} \cdot Y,
\end{align*}
where $\mathrm{Beta}(\alpha, \beta)$ denotes the beta distribution with shapes
$\alpha, \beta > 0$. From \citet{Neufeld2023}, we know that $Y^*, Y^\dagger$ are
independent, with $Y^* \sim \mathrm{Gam}(\epsilon, \epsilon \lambda)$ and
$Y^\dagger \sim \mathrm{Gam}(1-\epsilon, (1-\epsilon) \lambda)$, where 
$\mathrm{Gam}(\alpha, \beta)$ is the gamma distribution with shape 
$\alpha>0$ and rate $\beta>0$. Thus, we can see that $\E[Y^*] = \E[Y^\dagger]$,
and so (i) and (ii) are clearly satisfied. Furthermore, the distribution
$\mathrm{Gam}(\epsilon, \epsilon \lambda)$ of $Y^*$ is indeed similar to that  
$\mathrm{Exp}(\lambda)$ of $Y$, with the latter approaching the former as
$\epsilon \to 1$, which confirms our property (iii).  

Given the success we have seen for the CB method in the Gaussian and Poisson
settings, we feel these and other extensions are worth exploring, along of
course with theory and experiments to support their use as potentially core
tools for error and risk estimation in denoising and fixed-X regression
problems.    

\section*{Acknowledgements}

We thank Aaditya Ramdas and Boyan Duan for inspiring discussions. NLO was
supported by an Amazon Fellowship.   

\bibliographystyle{plainnat}
\bibliography{ryantibs, cbpois}

\newpage
\appendix

\section{Proofs and additional details for Section \ref{sec:cb}}

\subsection{Proof of Lemma \ref{lem:pois_binom}}
\label{app:pois_binom}

The joint probability mass function of $(Y,\omega)$ is given by
\begin{align*}
\P(Y = y, \omega = w) &= \P(\omega = k \,|\, Y = y)\P(Y = y) \\
&= \binom{y}{k} p^k(1-p)^{y-k} \frac{\mu^ye^{-\mu}}{y!}, 
\end{align*}
for all $y \geq 0$ and $k \in \{0, \dots, y\}$. The probability mass function of
$(U,V) = (Y-\omega, \omega)$ is thus
\begin{align*}
\P(U = u, V = v) &= \P(Y = u+v, \omega = v) \\
&= \binom{u+v}{v} p^v(1-p)^{u}\frac{\mu^{u+v}e^{-\mu}}{(u+v)!}\\
&= \frac{1}{u!v!} p^v(1-p)^{u}\mu^u \mu^ve^{-(p+1-p)\mu}\\
&= \frac{((1-p)\mu)^ue^{-(1-p)\mu}}{u!}\frac{(p\mu)^ve^{-p\mu}}{v!}.  
\end{align*}
This shows that $U$ and $V$ are independent $\Pois((1-p)\mu)$ and $\Pois(p\mu)$
random variables, respectively, which proves the desired result.

\subsection{Divergence deviance terms in UE versus CB}
\label{app:diverging_dev}

For deviance loss, a summand in the unbiased estimator \eqref{eq:ue_dev} is
undefined if $Y_i \not= 0$ and $g_i(Y-e_i)=0$, and a summand in the CB estimator
\eqref{eq:cb_dev} is undefined if \smash{$Y^\dagger_i \not= 0$} and
\smash{$g_i(Y^*) = 0$}, where $Y^*,Y^\dagger$ are a sample from 
\eqref{eq:cb_noise} (we have hidden the dependence on $b$). Suppose for
simplicity that $n=1$ and $g(0)=0$. We can then compute, under $Y \sim 
\Pois(\mu)$, the probability with which this happens for the unbiased and CB
estimators. For the unbiased estimator, this is:   
\[
\P(Y = 1) = e^{-\mu} \mu.
\]
For the CB estimator, this is:
\begin{align*}
\sum_{y=1}^\infty \P(Y^* = 0 \,|\, Y = y) \P(Y = y) 
&= \sum_{y=1}^\infty \P(\omega = Y \,|\, Y = y) \P(Y = y) \\ 
&= \sum_{y=1}^\infty p^y e^{-\mu} \mu^y/y! \\
&= e^{-(1-p)\mu} \sum_{y=1}^\infty e^{-p\mu} (p\mu)^y/y! \\
&= e^{-(1-p)\mu} (1-e^{-p\mu}) \\
&= e^{-\mu}(e^{p\mu}-1).
\end{align*}
Figure \ref{fig:diverging_dev} plots these two probabilities, a functions of
$\mu$, for $p$ ranging over 0.01, 0.1, 0.3, 0.5. For small $p$, it is clear that
the CB estimator has much lower probability of being ill-defined than the
unbiased estimator.     

\begin{figure}[htb]
\centering
\includegraphics[width=\textwidth]{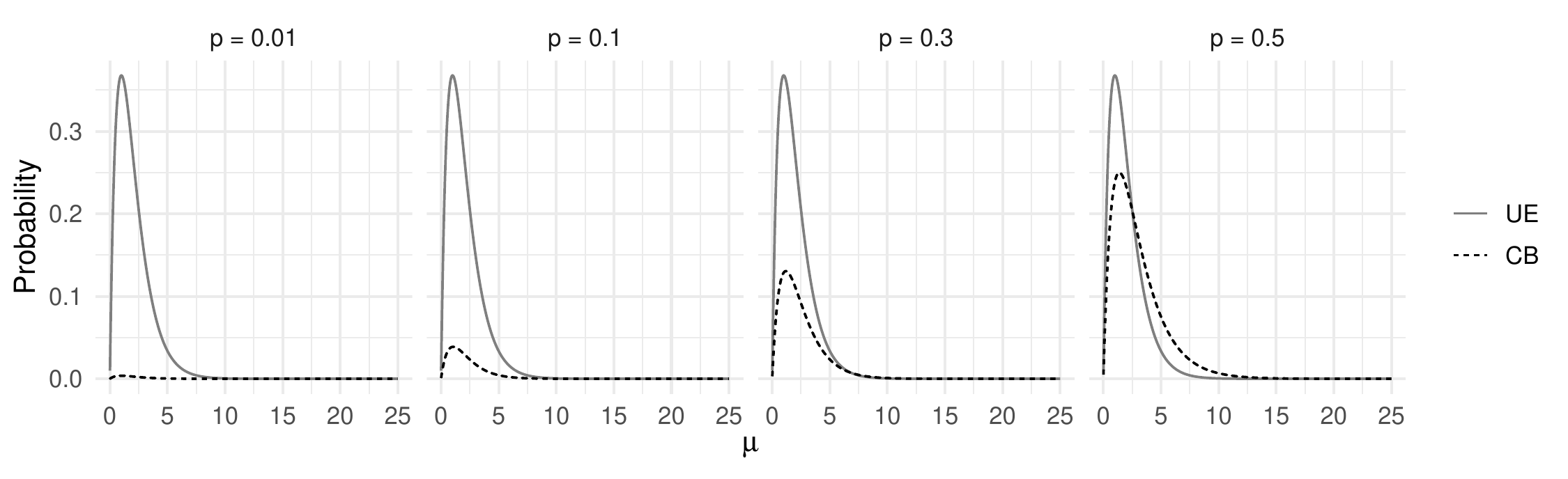}
\caption{Comparison of the probabilities of an individual summand from the 
  unbiased and CB estimators being ill-defined, as functions of the mean
  $\mu$. The four panels show different values of $p$.} 
\label{fig:diverging_dev}
\end{figure}

\allowdisplaybreaks
\subsection{Proof of Proposition \ref{prop:test_error_p_smoothness}}
\label{app:test_error_p_smoothness}

In this proof we use \smash{$f(Y_p, \newY_p)$} to denote \smash{$D_\phi(\newY_p,  
  g(Y_p))$}. First, we show that the map is continuous. For any $p \in [0,1)$,
\begin{align*}
\lim_{t \to p} \E[f(Y_t, \newY_t)] 
&= \lim_{t \to p} 
\sum_{y_1, \dots, y_n=0}^\infty \sum_{\newy_1, \dots, \newy_n=0}^\infty 
f(y, \newy)  \frac{e^{-2(1-t)\sum_{i=1}^n\mu_i} 
(\prod_{i=1}^n \mu_i^{y_i+\newy_i}) (1-t)^{\sum_{i=1}^n y_i+\newy_i}} 
{\prod_{i=1}^n y_i!\newy_i!} \\
&= \sum_{y_1, \dots, y_n=0}^\infty \sum_{\newy_1, \dots, \newy_n=0}^\infty  
f(y, \newy)  \frac{e^{-2(1-p)\sum_{i=1}^n\mu_i} 
(\prod_{i=1}^n \mu_i^{y_i+\newy_i}) (1-p)^{\sum_{i=1}^n y_i+\newy_i}} 
{\prod_{i=1}^n y_i!\newy_i!} \\
&= \E[f(Y_p, \newY_p)].
\end{align*}
To switch the infinite sum and the limit, we used the dominated convergence
theorem (DCT). The dominating function is given by
\[
h(\newy, y) = f(y, \newy) \frac{\prod_{i=1}^n \mu_i^{y_i+\newy_i}}
{\prod_{i=1}^n y_i!\newy_i!},
\]
which is integrable by assumption:
\begin{align*}
\sum_{y_1, \dots, y_n=0}^\infty \sum_{\newy_1, \dots, \newy_n=0}^\infty
 h(\newy, y) 
&=  \sum_{y_1, \dots, y_n=0}^\infty \sum_{\newy_1, \dots, \newy_n=0}^\infty 
f(y, \newy) \frac{\prod_{i=1}^n \mu_i^{y_i+\newy_i}}{\prod_{i=1}^n 
y_i!\newy_i!}e^{-2\sum_{i=1}^n\mu_i}e^{2\sum_{i=1}^n \mu_i} \\
&= e^{2\sum_{i=1}^n\mu_i}\E[f(Y_0, \newY_0)] < \infty.
\end{align*} 

Next, for the first derivative, note that
\begin{align*}
\frac{\partial}{\partial p} \E[f(Y_p, \newY_p)]
&= \frac{\partial}{\partial p} \sum_{y_1, \dots, y_n=0}^\infty 
\sum_{\newy_1, \dots, \newy_n=0}^\infty f(y, \newy)  
\frac{e^{-2(1-p)\sum_{i=1}^n\mu_i} \prod_{i=1}^n \mu_i^{y_i+\newy_i} 
(1-p)^{\sum_{i=1}^n y_i+\newy_i}}{\prod_{i=1}^n y_i!\newy_i!} \\
&= \sum_{y_1, \dots, y_n=0}^\infty \sum_{\newy_1, \dots, \newy_n=0}^\infty 
\frac{f(y, \newy)\prod_{i=1}^n \mu_i^{y_i+\newy_i}}
{\prod_{i=1}^n y_i!\newy_i!}\frac{\partial}{\partial p}  
e^{-2(1-p)\sum_{i=1}^n\mu_i}(1-p)^{\sum_{i=1}^n y_i+\newy_i} \\
&= \sum_{y_1, \dots, y_n=0}^\infty \sum_{\newy_1, \dots, \newy_n=0}^\infty 
\frac{f(y, \newy)\prod_{i=1}^n \mu_i^{y_i+\newy_i}}
{\prod_{i=1}^n y_i!\newy_i!} \bigg( 2\sum_{i=1}^n\mu_i
e^{-2(1-p)\sum_{i=1}^n\mu_i}(1-p)^{\sum_{i=1}^ny_i+\newy_i} \\
& \qquad - \sum_{i=1}^n (y_i+\newy_i)e^{-2(1-p)\sum_{i=1}^n\mu_i}
(1-p)^{\sum_{i=1}^ny_i+\newy_i -1} \bigg) \\
&= 2\sum_{i=1}^n\mu_i \E[f(Y_p, \newY_p)] -\frac{1}{1-p} 
\E\big[ f(Y_p, \newY_p)\langle \newY_p+Y_p, 1_n \rangle\big], 
\end{align*}
where we used DCT to switch the sums and derivative, using a similar dominating 
function as above and recognizing that the summand is Lipschitz in $p$ with
Lipschitz constant depending on \smash{$\mu,y,\newy$}. Now to prove continuity  
of the first derivative, we apply the above continuity result with
\smash{$f(Y_p, \newY_p)\langle \newY_p+Y_p, 1_n \rangle$} in place of
\smash{$f(Y_p, \newY_p)$}. For $k\th$ derivatives, the argument follows from
sequential applications of the same continuity result and similar derivative
calculations.  

\section{Proof of Theorem \ref{thm:noiseless_limit}}
\label{app:noiseless_limit}

We start with a lemma that contains two key results to be used in the proof of
Theorem \ref{thm:noiseless_limit}.      

\begin{lemma}
\label{lem:limiting_inner_prod} 
Let $h: \Z^n_+ \to \R^n_+$, and set $h(z) = 0$ for $z \notin \Z^n_+$. Fix any $y
\in \Z_+^n$, and draw $\omega_i \sim \Binom(y_i, p)$, independently, for
$i=1,\dots,n$, where $p \in [0,1)$. Then, for each $i=1,\dots,n$,
\begin{enumerate}[(a)]
\item $\lim_{p \to 0} \E[h_i(y-\omega)] = h_i(y)$;
\item $\lim_{p \to 0} \frac{1-p}{p} \cdot \E[\omega_i h_i(y-\omega)] =    
  y_i h_i(y-e_i)$.
\end{enumerate}
Recall, we use $e_i \in \R^n$ to denote the vector whose $i\th$ entry is 1, with 
all others 0. 
\end{lemma}

\begin{proof}
Define the following sets:
\begin{align*}
\Omega &= \{(\omega_1, \dots, \omega_n):
\omega_j \in \{0, \dots, y_j\}, j = 1,\dots,n\}, \\
\Omega_{\setminus 0} &= \Omega \setminus \{0\}, \\
\Omega_{i0} &= \{(\omega_1, \dots, \omega_n):
\omega_i = 0, \omega_j \in \{0, \dots, y_j\}, j \not= i\}, \\
\Omega_{i10} &= \{(\omega_1, \dots, \omega_n):
\omega_i = 1, \omega_j = 0, j \not=i \}, \\
\Omega_{i1\bar 0}&= \{(\omega_1, \dots, \omega_n):
\omega_i \geq 1, \omega_j \in \{0, \dots, y_j\},
j \not= i\} \setminus \Omega_{i10}.
\end{align*}
For the first result (a), we have 
\begin{align*}
\lim_{p\to 0} \E[h_i(y-\omega)] 
&= \lim_{p\to 0} \sum_{\omega \in \Omega} h_i(y-\omega)
p^{\sum_{k=1}^n\omega_k} (1-p)^{\sum_{k=1}^n y_k-\omega_k}
\prod_{j=1}^n\binom{y_j}{\omega_j} \\
&= \lim_{p\to 0} h_i(y) (1-p)^{\sum_{k=1}^n y_k} \\
&\qquad + \sum_{\omega \in \Omega_{\setminus 0}} 
\lim_{p\to 0} h_i(y-\omega)p^{\sum_{k=1}^n\omega_k} 
(1-p)^{\sum_{k=1}^n y_k-\omega_k}
\prod_{j=1}^n\binom{y_j}{\omega_j} \\
&= h_i(y),
\end{align*}
since for $\omega \in \Omega_{\setminus 0}$, we have that
\smash{$\sum_{k=1}^n\omega_k > 0$}. For the second result (b), 
\begin{align*}
\lim_{p\to 0} \frac{1-p}{p} \E[\omega_i h_i(y-\omega)] 
&= \lim_{p\to 0} \frac{1-p}{p} \sum_{\omega \in \Omega} 
\omega_i h_i(y-\omega) p^{\sum_{k=1}^n\omega_k} 
(1-p)^{\sum_{k=1}^n y_k-\omega_k}
\prod_{j=1}^n\binom{y_j}{\omega_j} \\
&= \lim_{p\to 0} \sum_{\omega \in \Omega} 
\omega_i h_i(y-\omega)p^{\sum_{k=1}^n\omega_k -1}
(1-p)^{1+\sum_{k=1}^n y_k-\omega_k}\prod_{j=1}^n
\binom{y_j}{\omega_j} \\
&= \lim_{p\to 0} \sum_{\omega \in \Omega_{i1 0}} 
\omega_i h_i(y-\omega)p^{\sum_{k=1}^n\omega_k -1}
(1-p)^{1+\sum_{k=1}^n Y_k-\omega_k}
\prod_{j=1}^n\binom{y_j}{\omega_j} \\
&= \lim_{p\to 0} h_i(y-e_i)p^0 (1-p)^{1+\sum_{k=1}^n y_k-\omega_k}
\binom{y_i}{1} \prod_{j\neq i, j=1}^n\binom{y_j}{0} \\
&= y_ih_i(y-e_i), 
\end{align*}
where we use the fact that $\Omega = \Omega_{i0} \cup \Omega_{i10} \cup
\Omega_{i1\bar 0}$ and 
\begin{multline*}
\lim_{p\to 0} \sum_{\omega \in \Omega_{i0}} 
\omega_i h_i(y-\omega)p^{\sum_{k=1}^n\omega_k -1} 
(1-p)^{1+\sum_{k=1}^n y_k-\omega_k}\prod_{j=1}^n\binom{y_j}{\omega_j} \\
= \lim_{p\to 0} \sum_{\omega \in \Omega_{i0}}  0 h_i(y-\omega)
p^{\sum_{k=1}^n\omega_k -1} (1-p)^{1+\sum_{k=1}^n 
y_k-\omega_k}\prod_{j=1}^n\binom{y_j}{\omega_j} = 0
\end{multline*}
as well as 
\begin{multline*}
\lim_{p\to 0} \sum_{\omega \in \Omega_{i1\bar 0}} 
\omega_i h_i(y-\omega)p^{\sum_{k=1}^n\omega_k -1} 
(1-p)^{1+\sum_{k=1}^n y_k-\omega_k}\prod_{j=1}^n\binom{y_j}{\omega_j} \\
= \sum_{\omega \in \Omega_{i1\bar 0}}  \omega_i h_i(y-\omega) 
0^{\sum_{k=1}^n\omega_k -1} \prod_{j=1}^n\binom{y_j}{\omega_j} = 0,
\end{multline*}
since for $\omega \in \Omega_{i1\bar 0}$, we have that
\smash{$\sum_{k=1}^n\omega_k - 1 >0$}. 
\end{proof}

Now we are ready to prove Theorem \ref{thm:noiseless_limit}. We start by
expanding the infinite bootstrap estimator  
\begin{align*}
\E[\CB_p(g) \,|\, Y]
&= \sum_{i=1}^n \E\Big[ D_\phi(Y^\dagger_i, g_i(Y^*)) + \phi(Y^*_i)
-\phi(Y^\dagger_i) \,\big|\, Y \Big] \\
&= \sum_{i=1}^n \E\Big[ \phi(Y^*_i) - \phi(g_i(Y^*)) - 
\nabla_i\phi(g(Y^*))(Y^\dagger_i - Y^*_i) \,\big|\, Y \Big] \\
&= \sum_{i=1}^n \E\bigg[ \phi(Y_i - \omega_i)-\phi(g_i(Y-\omega)) - 
\nabla_i\phi(g(Y - \omega)) \bigg( \frac{1-p}{p}\omega_i - (Y_i - \omega_i)
\bigg) \,\Big|\, Y \bigg] \\
&= \sum_{i=1}^n \E\bigg[ \phi(Y_i - \omega_i)-\phi(g_i(Y-\omega)) - 
\frac{1-p}{p}\nabla_i\phi(g(Y - \omega))\omega_i + 
\nabla_i\phi(g(Y - \omega)) (Y_i - \omega_i) \,\Big|\, Y \bigg].
\end{align*}
Then, taking the limit in the last line,
\begin{align*}
\sum_{i=1}^n \lim_{p\to 0} \,
& \E\bigg[ \phi(Y_i - \omega_i)-\phi(g_i(Y-\omega)) -  
\frac{1-p}{p}\nabla_i\phi(g(Y - \omega))\omega_i + 
\nabla_i\phi(g(Y - \omega)) (Y_i - \omega_i) \,\Big|\, Y \bigg] \\
&= \sum_{i=1}^n \lim_{p\to 0} 
\E \Big[ \phi(Y_i - \omega_i)-\phi(g_i(Y-\omega)) + 
\nabla_i\phi(g(Y - \omega)) (Y_i - \omega_i) \,\big|\, Y \Big] - 
Y_i\nabla_i\phi(g(Y - e_i)) \\
&= \sum_{i=1}^n \phi(Y_i)-\phi(g_i(Y)) + \nabla_i\phi(g(Y)) (Y_i ) - 
Y_i\nabla_i\phi(g(Y - e_i)) \\
&= \UE(Y),
\end{align*}
where in the second-to-last and last lines we used Lemma 
\ref{lem:limiting_inner_prod} parts (b) and (a), respectively. This completes
the proof. 

\section{Proofs for Section \ref{sec:bias_variance}}
 
\subsection{Proof of Proposition \ref{prop:cb_bias}}
\label{app:cb_bias}

From Proposition \ref{prop:test_error_p_smoothness}, the mapping $p \mapsto
\Err_p(g)$ has a continuous derivative for $p \in [0,1)$. By an application  
of the fundamental theorem of calculus, we can write the bias as 
\begin{align*}
\Err_p(g) &- \Err(g) = \int_0^p \frac{\partial}{\partial t}  \Err_t(g) \, dt \\   
&= \int_0^p \bigg\{ 2\sum_{i=1}^n\mu_i \E[D_\phi(\newY_t, g(Y_t))] -
\frac{1}{(1-t)} \E\big[ D_\phi(\newY_t, g(Y_t))\langle \newY_t+Y_t, 1_n
\rangle\big] \bigg\} \, dt \\  
&= \int_0^p\bigg\{ 2\sum_{i=1}^n \mu_i \Err_t(g) -\frac{1}{(1-t)} 
\Cov\Big( D_\phi(\newY_t, g(Y_t)),\langle \newY_t+Y_t, 1_n \rangle\Big)
-\frac{1}{(1-t)} 2\Err_t(g) (1-t)\sum_{i=1}^n\mu_i \bigg\} \,dt \\
&= -\int_0^p \frac{1}{(1-t)} \Cov\Big( D_\phi(\newY_t, g(Y_t)),\langle
\newY_t+Y_t, 1_n \rangle\Big) \, dt \\
&= -\int_0^p \frac{1}{(1-t)} \Cor\Big( D_\phi(\newY_t, g(Y_t)),\langle
\newY_t+Y_t,  1_n \rangle\Big) \sqrt{\Var\big[ D_\phi(\newY_t, g(Y_t)) \big]} 
\sqrt{\Var\big[ \langle \newY_t+Y_t, 1_n \rangle\big]} \, dt \\
&= -\int_0^p \frac{1}{(1-t)} \Cor\Big( D_\phi(\newY_t, g(Y_t)), 
\langle \newY_t+Y_t, 1_n \rangle\Big) 
\sqrt{\Var\big[ D_\phi(\newY_t, g(Y_t)) \big]} 
\sqrt{(2(1-t)\sum_{i=1}^n\mu_i} \, dt \\ 
&= -\int_0^p \frac{1}{\sqrt{1-t}} \Cor\big[ D_\phi(\newY_t, g(Y_t)),
\langle \newY_t+Y_t, 1_n \rangle\Big)
\sqrt{\Var\big[ D_\phi(\newY_t, g(Y_t)) \big]} \, dt 
\cdot \sqrt{2\sum_{i=1}^n\mu_i}. 
\end{align*}
which proves \eqref{eq:cb_bias}. Upper bounding the correlation by $1$, and
using the monotone variance assumption for $p \in [0,1/2]$, we have  
\begin{align*}
|\Err_p(g) - \Err(g)| &\leq \int_0^p \frac{1}{\sqrt{1-t}} 
\sqrt{\Var\big[ D_\phi(\newY_t, g(Y_t)) \Big)} \, dt 
\cdot \sqrt{2\sum_{i=1}^n\mu_i} \\
&\leq \int_0^p \frac{1}{\sqrt{1-t}} 
\sqrt{\Var\big[ D_\phi(\newY, g(Y)) \Big)} \, dt
\cdot \sqrt{2\sum_{i=1}^n\mu_i} \\
&= \sqrt{\Var\big[ D_\phi(\newY, g(Y)) \Big)} 
\sqrt{2\sum_{i=1}^n\mu_i}
\int_0^p \frac{1}{\sqrt{1-t}} \, dt \\
&= \sqrt{\Var\big[ D_\phi(\newY, g(Y)) \Big)} 
\sqrt{2\sum_{i=1}^n\mu_i}
\frac{2p}{1+\sqrt{1-p}} \\
&= \sqrt{\Var\big[ D_\phi(\newY, g(Y)) \Big)}
\sqrt{2\sum_{i=1}^n\mu_i} \frac{5p}{3}, 
\end{align*}
which proves \eqref{eq:cb_bias_bd}.

\subsection{Proof of Proposition \ref{prop:cb_rvar}}
\label{app:cb_rvar}

We start from the fact that
\[
\Var[\CB_p \,|\, Y] = \frac{1}{B} \Var\Big[ D_\phi(Y^\dagger, g(Y^*)) +
\phi(Y^*) - \phi(Y^\dagger) \,\big|\, Y \Big].
\]
Then
\begin{align*}
\E\big[ \Var[\CB_p \,|\, Y] \big]
&= \frac{1}{B} \E\bigg[ \Var\Big[ D_\phi(Y^\dagger, g(Y^*)) + \phi(Y^*) -
\phi(Y^\dagger) \,\big|\, Y \Big]\bigg] \\
&= \frac{1}{B} \E\bigg[ \Var\Big[ \phi(Y^*) - \phi(g(Y^*)) -
\langle\nabla\phi(g(Y^*)),Y^\dagger-g(Y^*) \rangle \,\big|\,Y \Big]\bigg] \\
&= \frac{1}{B} \E\bigg[ \Var\Big[ \phi(Y^*) - \phi(g(Y^*)) - 
\langle\nabla\phi(g(Y^*)),Y^\dagger - Y^* + Y^* -g(Y^*) \rangle \,\big|\,Y
\Big]\bigg] \\
&= \frac{1}{B} \E\bigg[ \Var\Big[ D_\phi(Y^*,g(Y^*)) -
\langle\nabla\phi(g(Y^*)), Y^\dagger - Y^*\rangle \,\big|\,Y \Big]\bigg] \\
&= \frac{1}{B} \E\bigg[ \Var\Big[ D_\phi(Y^*,g(Y^*)) - 
\frac{1}{p}\langle\omega, \nabla\phi(g(Y^*))\rangle + 
\langle\nabla\phi(g(Y^*)),Y\rangle \,\big|\,Y \Big]\bigg] \\
&\leq \frac{2}{B} \E\bigg[ \Var\Big[ D_\phi(Y^*,g(Y^*)) + 
\langle\nabla\phi(g(Y^*)),Y\rangle \,\big|\,Y \Big]\bigg] + 
\frac{2}{Bp^2}\E\bigg[ \Var\Big[ \langle\omega, \nabla\phi(g(Y^*))
\rangle \,\big|\,Y \Big]\bigg] \\
&\leq \frac{2}{B}\Var\Big[ D_\phi(Y^*,g(Y^*)) +
\langle Y^*,\nabla\phi(g(Y^*))\rangle \Big] + 
\frac{2}{Bp^2}\Var\Big[ \langle\omega, \nabla\phi(g(Y^*))\rangle \Big] \\
&= \frac{2}{B}\Var\Big[ D_\phi(Y,g(Y)) +\langle Y,\nabla\phi(g(Y))\rangle \Big]
+  \frac{2}{B}O(p) + \frac{2}{Bp^2}
\Var\big[\langle\omega, \nabla\phi(g(Y^*)) \rangle\big],
\end{align*}
where the second-to-last line uses the law of total variance, and the last line
uses a Taylor expansion which follows from the continuity in $p$ of expected
value of functions of $Y_p$ as assumed in the proof of Proposition
\ref{prop:test_error_p_smoothness}. For the last term in the above display,
note that $\omega$ and $Y^*$ are independent. Therefore, we can use that fact
that if $X_1$ and $X_2$ are independent, then $\Var[X_1X_2] = \Var[X_1]\Var[X_2]
+ \Var[X_1]\E^2[X_2] + \Var[X_2]\E^2[X_1]$, which translates to
\begin{align*}
\frac{2}{Bp^2}\Var\big[ &\langle \omega, \nabla\phi(g(Y^*)) \rangle\big] \\
&= \frac{2}{Bp^2}\sum_{i=1}^n\Var\big[ \omega_i \nabla_i\phi(g(Y^*))\big] \\    
&= \frac{2}{Bp^2}\sum_{i=1}^n\Big( 
\Var[\omega_i] \Var\big[ \nabla_i\phi(g(Y^*)) \big] + 
\Var[\omega_i] \E^2\big[ \nabla_i\phi(g(Y^*)) \big] + 
\Var\big[ \nabla_i\phi(g(Y^*)) \big] \E^2[\omega_i] \Big) \\
&= \frac{2}{Bp^2}\sum_{i=1}^n\Big( 
p \mu_i \Var\big[ \nabla_i\phi(g(Y^*)) \big] +
p \mu_i \E^2\big[ \nabla_i\phi(g(Y^*)) \big] + 
\Var\big[ \nabla_i\phi(g(Y^*)) \big] p^2 \mu_i^2 \Big) \\
&= \frac{2}{Bp}\sum_{i=1}^n \mu_i \E\big[ \nabla_i\phi(g(Y^*))^2 \big] +  
\frac{2}{B}\sum_{i=1}^n \mu_i^2\Var\big[ \nabla_i\phi(g(Y^*)) \big] \\
&= \frac{2}{Bp}\sum_{i=1}^n \mu_i \E\big[ \nabla_i\phi(g(Y^*))^2 \big].
\end{align*}
Putting it all together gives the desired result \eqref{eq:cb_rvar_bd}.

\subsection{Proof of Proposition \ref{prop:cb_ivar}}
\label{app:cb_ivar}

Note that the irreducible variance $\Var(\E[\CB_p(g) \,|\, Y])$ does not depend
on $B$ (because the inner expectation does not), so we assume without a  
loss of generality that $B=1$ henceforth. Observe that
\begin{align}
\nonumber
\Var\big( &\E[\CB_p(g) \,|\, Y] \big) \\
\nonumber
&= \Var\Big( \E\Big[ D_\phi(Y^*, g(Y^*)) + 
\langle \nabla \phi(g(Y^*)), Y^* -  Y^\dagger\rangle \,\big|\, Y \Big]\Big) \\
\nonumber
&= \Var\bigg( \E\bigg[ D_\phi(Y^*, g(Y^*)) + 
\langle \nabla\phi(g(Y^*)), Y^*\rangle - \frac{1-p}{p}
\langle \nabla\phi(g(Y^*)), \omega\rangle \,\Big|\, Y \bigg]\bigg) \\
\label{eq:last_step}
&\leq 2\Var\Big( \E\Big[ D_\phi(Y^*, g(Y^*)) + 
\langle \nabla\phi(g(Y^*)), Y^*\rangle \big|\, Y \Big]\Big) + 
2\Var\bigg( \E\bigg[ \frac{1-p}{p}\langle \nabla\phi(g(Y^*)), 
\omega\rangle \Big|\, Y \bigg]\bigg).
\end{align}
For the first term in \eqref{eq:last_step}, we apply the law of total variance and 
\begin{align*}
\Var\Big( \E\Big[ D_\phi(Y^*, g(Y^*)) + \langle \nabla\phi(g(Y^*)), Y^*\rangle
  \Big] \Big) 
&\leq \Var\Big[ D_\phi(Y^*, g(Y^*)) + \langle \nabla\phi(g(Y^*)), Y^*\rangle
\Big] \\
&\to \Var\Big[ D_\phi(Y, g(Y)) + \langle \nabla\phi(g(Y)), Y\rangle\Big], 
\quad \text{as $p \to 0$},
\end{align*}
where the last convergence is guaranteed provided that $D^2_\phi(Y,g(Y))$ and
$\langle \nabla\phi(g(Y)), Y\rangle^2$ have finite expectation, according to the
first step in the proof of Proposition \ref{prop:test_error_p_smoothness}. 

For the second term in \eqref{eq:last_step} recalling that $\Phi_g$ as defined
in the statement of the proposition, 
\begin{align*}
\lim_{p\to 0} \Var\bigg( \E\bigg[ \frac{1-p}{p}\langle \nabla\phi(g(Y^*)),
  \omega\rangle \,\Big|\, Y \bigg]\bigg) 
&\leq \lim_{p\to 0} \E\bigg( \E^2\bigg[ \frac{1-p}{p}\langle \nabla\phi(g(Y^*)),
\omega\rangle \,\Big|\, Y \bigg]\bigg) \\
&\leq \lim_{p\to 0}\E\bigg( \E^2\bigg[ \frac{1-p}{p}\langle \Phi_g(Y), \omega 
\,\Big|\, Y \bigg]\bigg) \\
&= \lim_{p\to 0} \, (1-p)^2\E\Big[ \langle \Phi_g(Y), \E[\omega/p \,|\, Y] 
\rangle^2\Big] \\ 
&= \lim_{p\to 0} \, (1-p)^2\E\big[ \langle \Phi_g(Y), Y \rangle^2 \big] \\ 
&=\E\big[ \langle \Phi_g(Y), Y\rangle^2 \big].
\end{align*}
Putting it all together gives the desired result \eqref{eq:cb_ivar_bd}.

\end{document}